%% file: main.tex
\documentclass{lmcs}
\pdfoutput=1

\usepackage{lastpage}
\lmcsdoi{15}{1}{30}
\lmcsheading{}{\pageref{LastPage}}{}{}%
{Nov.~01,~2017}{Aug.~23,~2022}{}

\usepackage{amssymb,latexsym,amsmath,bbm}
\usepackage{hyperref}

\input{preamble.tex}

\input{macros.tex}
\input{abbreviations-yv.tex}
\input{abbreviations-se.tex}

\title{Disjunctive Bases: \texorpdfstring{\\}{}Normal Forms and Model Theory for Modal Logics}


\date{}

\begin{document}

\author[S.~Enqvist]{Sebastian Enqvist\rsuper{a}}	%required
\address{\lsuper{a}Stockholm University, Department of Philosophy}	%required
\email{sebastian.enqvist@philosophy.su.se}  %optional
%\thanks{thanks 1, optional.}	%optional

\author[Y.~Venema]{Yde Venema\rsuper{b}}	%optional
\address{\lsuper{b}University of Amsterdam, Institute for Logic, Language and Computation}	%optional
\email{y.venema@uva.nl}  %optional

\begin{abstract}
We present the concept of a disjunctive basis as a generic framework for normal forms in modal
logic based on coalgebra.
Disjunctive bases were defined in previous work on completeness for modal fixpoint logics, where
they played a central role in the proof of a generic completeness theorem for coalgebraic mu-calculi.
Believing the concept has a much wider significance,
% especially in relation to modal fixpoint logics, and here,
here we investigate it more thoroughly in its own right.
We show that the presence of a disjunctive basis at the ``one-step'' level entails a number of good
properties for a coalgebraic mu-calculus, in particular, a simulation theorem showing that every
alternating automaton can be transformed into an equivalent nondeterministic one.
Based on this, we prove a Lyndon theorem for the full fixpoint logic, its fixpoint-free
fragment and its one-step fragment, and a Uniform Interpolation result, for both the full mu-calculus
and its fixpoint-free fragment.

We also raise the questions, when a disjunctive basis exists, and how disjunctive bases are
related to Moss' coalgebraic ``nabla'' modalities.
Nabla formulas provide disjunctive bases for many coalgebraic modal logics, but there are
cases where disjunctive bases give useful normal forms even when nabla formulas fail to do so,
our prime example being graded modal logic. We also show that disjunctive bases are preserved by forming sums, products and compositions of coalgebraic modal logics, providing  tools for modular construction of modal logics admitting disjunctive bases.
Finally, we consider the problem of giving a category-theoretic formulation of disjunctive bases, and
provide a partial solution.

This is a corrected version of the paper \url{https://arxiv.org/abs/1710.10706} published
originally on 26/3, 2019.
\end{abstract}

\maketitle

\input{sec-intro}
\input{sec-preliminaries}
\input{sec-dis-bases}
\input{sec-simulation}
\input{sec-lyndon}
\input{sec-unifitp}
%\input{sec-jw-theorem}
\input{sec-combined}
\input{sec-yoneda}

\subsubsection*{Acknowledgement}
We thank the two anonymous referees for their detailed and helpful comments.

%%
%% Bibliography
%%

%% Either use bibtex (recommended),

\bibliographystyle{plain}
\bibliography{ml-book,nabla,mu,extra}

%% .. or use the thebibliography environment explicitely

% \newpage
\appendix
%\insertt{neutral/sec-appendix}

\end{document}

%% file: preamble.tex
\usepackage{graphicx}
\usepackage{amssymb,latexsym,amsmath,bbm}
\usepackage[all]{xy}
\usepackage{multicol}
\usepackage{url,color}

\newenvironment{tbs}{%
   \small\tt
   \begin{itemize}}{\end{itemize}}
\newcommand{\btbs}{\begin{tbs}}
\newcommand{\etbs}{\end{tbs}}

\newcommand{\hide}[1]{}

\theoremstyle{definition}
\newtheorem{definition}{Definition}[section]
\newtheorem{myfact}[definition]{Fact}

\newtheorem{remark}[definition]{Remark}

\theoremstyle{defC}
\newtheorem{myfactC}[definition]{Fact}

\newenvironment{proofof}[1]{\begin{trivlist}\item[\hskip\labelsep{\bf
Proof~of~{#1}.\ }]}{\hspace*{\fill} {\sc \qed}\end{trivlist}}

\newtheorem{claim2}{{\normalfont\sc Claim}\normalfont}
\newenvironment{claim}{\begin{claim2}\rm}{\end{claim2}\rm}
\newenvironment{claimfirst}{\setcounter{claim2}{0}
               \begin{claim2}\rm}{\end{claim2}\rm}

\newenvironment{pfclaim}{\begin{trivlist}\item[]{\sc Proof of
Claim}}{\hfill {\mbox{$\blacktriangleleft$}}\end{trivlist}}

%% file: macros.tex
\newcommand{\Set}{\mathsf{Set}}

\newcommand{\id}{\mathsf{id}}

\newcommand{\funP}{\mathsf{P}}
   
\newcommand{\funT}{\mathsf{T}}
   
   \newcommand{\funTX}{\mathsf{T}_{\Prop}}
   \newcommand{\fun}{\funT}
\newcommand{\funI}{\mathsf{Id}}
\newcommand{\funB}{\mathsf{B}}

\newcommand{\funM}{\mathsf{M}}

\newcommand{\funQ}{\breve{\funP}}

\newcommand{\cmp}{\circledcirc}

\newcommand{\nada}{\varnothing}  % empty set

\newcommand{\Dom}{\mathsf{Dom}}
\newcommand{\Ran}{\mathsf{Ran}}

\newcommand{\Id}{\mathsf{Id}}

\newcommand{\isdef}{\mathrel{:=}}

\newcommand{\rst}[1]{\!\!\upharpoonright_{#1}\,}
\newcommand{\sz}[1]{|#1|}

\newcommand{\NBT}{\mathit{NBT}}
\newcommand{\last}{\mathsf{last}}

\newcommand{\fl}[1]{#1^{\flat}}

\newcommand{\Sig}{\Sigma}

\newcommand{\Prop}{\ensuremath{\mathtt{X}}}        %PROP
\newcommand{\PropQ}{\mathtt{Y}}

\newcommand{\ML}{\ensuremath{\mathtt{ML}}} %ML
\newcommand{\MLLa}{\ensuremath{\ML_{\La}}}    %muML
\newcommand{\muML}{\ensuremath{\mu\ML}}    %muML
\newcommand{\muMLLa}{\ensuremath{\mu\ML_{\La}}}    %muML

\newcommand{\Latt}{\mathtt{Latt}}          %Latt
\newcommand{\MLone}{\ensuremath{\mathtt{1ML}}} %1ML
   \newcommand{\MLoLa}{\ensuremath{\mathtt{1ML}_{\La}}} %1ML
\newcommand{\D}{\mathtt{D}}

\newcommand{\SigB}{\Sig_{\funB}}
\newcommand{\SigM}{\Sig_{\funM}}

\newcommand{\dia}{\Diamond}
\newcommand{\nb}{\nabla}
\newcommand{\hs}{\heartsuit}
\newcommand{\mop}[1]{\hs_{#1}}

\newcommand{\nxt}[1]{\raisebox{.3ex}{$\scriptstyle \bigcirc$}_{#1}}

\newcommand{\diff}{\mathsf{diff}}
\newcommand{\disbag}[2]{\nabla(#1 {;} #2)}

\renewcommand{\phi}{\varphi} % nicer \phi
\newcommand{\isbnf}{\mathrel{::=}}
\newcommand{\divbnf}{\mid}

\newcommand{\mathstr}[1]{\mathbb{#1}}

\newcommand{\bbA}{\mathstr{A}}

\newcommand{\bbS}{\mathstr{S}}

\newcommand{\bbX}{\mathstr{X}}

\newcommand{\simu}{\mathrel{\mathchoice%
{\raisebox{.3ex}{$\,
  \underline{\makebox[.7em]{$\rightarrow$}}\,$}}%
{\raisebox{.3ex}{$\,
  \underline{\makebox[.7em]{$\rightarrow$}}\,$}}%
{\raisebox{.2ex}{$\,
  \underline{\makebox[.5em]{\scriptsize$\rightarrow$}}\,$}}%
{\raisebox{.2ex}{$\,
  \underline{\makebox[.5em]{\scriptsize$\rightarrow$}}\,$}}}}

\newcommand{\beheq}{\simeq}

\newcommand{\sat}{\Vdash}
\newcommand{\satone}{\sat^{1}}
\newcommand{\satzero}{\sat^{0}}
\newcommand{\mng}[1]{[\![ #1 ]\!]}
\newcommand{\mngone}[1]{[\![ #1 ]\!]^{1}}
\newcommand{\mngzero}[1]{[\![ #1 ]\!]^{0}}

\newcommand{\sh}[1]{{#1}^{\sharp}}
\newcommand{\shA}{\sh{A}}
\newcommand{\simof}[1]{\ensuremath{\mathsf{sim}(#1)}}

\newcommand{\AG}{\mathcal{A}}

\newcommand{\eloi}{\exists}
\newcommand{\abel}{\forall}

\newcommand{\Win}{\mathsf{Win}}

\newcommand{\PM}[1]{\mathrm{PM}_{#1}}

\newcommand{\bbG}{\mathbb{G}}

\newcommand{\tup}[1]{\langle#1\rangle}
\newcommand{\ul}[1]{\underline{#1}}

%% file: abbreviations-yv.tex
\newcommand{\sse}{\subseteq}

\newcommand{\bv}{\bigvee}
\newcommand{\bw}{\bigwedge}

\newcommand{\bdual}[1]{#1^{\partial}}

\newcommand{\Si}{\Sigma}
\newcommand{\Ga}{\Gamma}
\newcommand{\De}{\Delta}
\newcommand{\Th}{\Theta}
\newcommand{\La}{\Lambda}
\newcommand{\Om}{\Omega}
\newcommand{\al}{\alpha}
\newcommand{\be}{\beta}
\newcommand{\de}{\delta}
\newcommand{\ga}{\gamma}

\newcommand{\ka}{\kappa}
\newcommand{\la}{\lambda}
\newcommand{\si}{\sigma}
\newcommand{\om}{\omega}

\newcommand{\ai}{a_{I}}

\newcommand{\ol}[1]{\overline{#1}}
\newcommand{\ti}[1]{\widetilde{#1}}
\newcommand{\wh}[1]{\widehat{#1}}

%% file: abbreviations-se.tex
\newcommand{\Bool}{\mathtt{Bool}}

\newcommand{\psf}{\funP}

\newcommand{\oneseq}{\equiv^1}

\newcommand{\MLoneLa}{\MLoLa}

\newcommand\utimes{ \mathbin{\ooalign{$\cup$\cr%
   \hfil\raise0.42ex\hbox{$\scriptscriptstyle\times$}\hfil\cr}} }
\newcommand{\overlap}{\mathcal{O}}

%% file: sec-intro.tex
\section{Introduction}%
\label{s:intro}

The topic of this paper connects modal $\mu$-calculi, coalgebra and automata.
The connection between the modal $\mu$-calculus, as introduced by
Kozen~\cite{koze:resu83}, and automata running on infinite objects, is
standard~\cite{grae:auto02}.
Many of the most fundamental results about the modal $\mu$-calculus have been
proved by making use of this connection, including
completeness of Kozen's axiom system~\cite{walu:comp00},
and model theoretic results like expressive completeness~\cite{jani:expr96},
uniform interpolation and a Lyndon theorem~\cite{dago:logi00}.

The standard modal $\mu$-calculus was generalized to  generic, coalgebraic modal
$\mu$-calculi~\cite{vene:auto06}, of which the modal basis was provided by Moss'
original coalgebraic modality~\cite{moss:coal99}, now known as the \emph{nabla}
modality.
From a meta-logical perspective, Moss' nabla logics and their fixpoint
extensions are wonderfully well-behaved.
For example, a generic completeness theorem for nabla logics by a uniform system
of axioms was established~\cite{kupk:comp12}, and this was recently extended
to the fixpoint extension of the finitary Moss logic~\cite{enqv:comp16}.
Most importantly, the automata corresponding to the fixpoint extension of Moss'
finitary nabla logic always enjoy a \emph{simulation theorem}, allowing
arbitrary coalgebraic automata to be simulated by \emph{non-deterministic} ones;
this goes back to the work of Janin \& Walukiewicz on
$\mu$-automata~\cite{jani:auto95}.
The simulation theorem provides a very strong normal form for these logics, and
plays an important role  in the proofs of several results for coalgebraic
fixpoint logics.

The downside of this approach is that the nabla modality is rather non-standard,
and understanding what concrete formulas actually say is not always easy.
For this reason, another approach to coalgebraic modal logic has become popular,
based on so called \emph{predicate liftings}.
This approach, going back to the work of Pattinson~\cite{patt:coal03},
provides a much more familiar syntax in concrete applications, but can still be
elegantly formulated at the level of generality and abstraction that makes the
coalgebraic approach to modal logic attractive in the first place.%
\footnote{%
    For a comparison between the two approaches, see~\cite{kurz:moda12}.
}
% Following up on the second author's paper~\cite{vene:auto06},
Coalgebraic $\mu$-calculi have also been developed as extensions of the
predicate liftings based languages~\cite{cirs:expt09}, and the resulting logics
are very well behaved: for example, good complexity results were
obtained in op.~cit.
Again, the connection between formulas and automata can be formulated in this
setting~\cite{font:auto10}, but a central piece is now missing: so far, no
simulation theorem has been established for logics based on predicate liftings.
In fact, it is not trivial even to define what a non-deterministic automaton
\emph{is} in this setting.

This problem turned up in recent work~\cite{enqv:comp17}, by ourselves together
with Seifan, where we extended our earlier completeness result for Moss-style
fixpoint logics~\cite{enqv:comp16} to the predicate liftings setting.
% Indeed, our main problem there was to provide a simulation theorem for
% coalgebraic automata.
Our solution was to introduce the concept of a \emph{disjunctive basis}, which
formalizes in a compact way the minimal requirements that a collection of
predicate liftings $\La$ must meet in order for the class of corresponding
$\La$-automata to admit a simulation theorem.
Our aim in the present paper is to follow up on this conceptual contribution,
which we believe is of much wider significance besides providing a tool to
prove completeness results.

Exemplifying this, we shall explore some of the applications of our coalgebraic
simulation theorem.
Some of these transfer known results for nabla based fixpoint logics to the
predicate liftings setting; for example, we show that a linear-size model
property holds for our non-deterministic automata  (or ``disjunctive'' automata
as we will call them), following~\cite{vene:auto06}.
We also show that uniform interpolation results hold for coalgebraic fixpoint
logics in the presence of a disjunctive basis, which was proved for the
Moss-style languages in~\cite{mart:unif15}.
We prove a Lyndon theorem for coalgebraic fixpoint logics, generalizing
a result for the standard modal $\mu$-calculus proved in~\cite{dago:logi00}: a
formula is monotone in one of its variables if and only if it is equivalent to
one in which the variable appears positively.
We also prove an explicitly \emph{one-step} version of this last result, which
we believe has some practical interest for modal fixpoint logics:
It is used to show that, given an expressively complete set of monotone
predicate liftings, its associated $\mu$-calculus has the same expressive power
as the full $\mu$-calculus based on the collection of all monotone predicate
liftings.
Finally, we show that the sum, product and composition of two modal logics that
admit disjunctive bases also admit a disjunctive basis, thus providing a tool
for modular construction of logics to which our results apply.

Next to proving these results, we compare the notion of a disjunctive basis to
the nabla based approach to coalgebraic fixpoint logics.
The connection will be highlighted in Section~\ref{s:yoneda} where we discuss
disjunctive predicate liftings via the Yoneda lemma: here the Barr lifting of
the ambient functor (on which the semantics of nabla modalities are based) comes
into the picture naturally.
However, this is not to say that disjunctive bases are just ``nablas in disguise'': it is a fundamental concept, formulated specifically to suit the approach to coalgebraic modal logics via predicate liftings, rather than nabla based languages. Furthermore, in cases where these two approaches to coalgebraic logics are not equivalent, disjunctive bases may be available even when nabla formulas fail to provide them.
We shall see that there are several concrete and natural examples of this.
A particularly interesting specimen is  \emph{graded modal logic}, which adds
counting modalities to standard modal logic.
While we will see that the standard nabla formulas of this language do not
provide a disjunctive basis, nevertheless a disjunctive basis for graded modal
logic does exist.

\paragraph{\textbf{Correction of earlier version:}}

This paper is a corrected version of the paper \url{https://arxiv.org/abs/1710.10706} published
originally on 26/3, 2019. That version included a Janin-Walukiewicz style characterization theorem for coalgebraic
$\mu$-calculi admitting a slightly stronger form of disjunctive bases called \emph{uniform disjunctive bases}, claiming that such $\mu$-calculi are expressively complete for the corresponding coalgebraic monadic second-order logics modulo bisimulation invariance (see ~\cite{enqv:mona15,enqv:expr17}). The proof of this result was incorrect. We currently do not know whether the result holds, and we leave this as an open question for future research.  We shall briefly explain the error here. 

The argument given in the original paper defined a construction that provides, for every $\funT$-model $(\bbS,s)$, a
pre-image $f_\bbS : (\bbS_*,s_*) \to (\bbS,s)$ such that any disjunctive $\Lambda$-automaton $\bbA$ accepts $(\bbS,s)$ iff 
  it strongly accepts $(\bbS_*,s_*)$. An inductively defined translation $t$ was then given, mapping each formula in the monadic second-order language $ \mathtt{MSO}_\Lambda(\Prop)$ to the language $ \muML_\Lambda(\Prop)$. The key claim about this translation was that, for any formula $\varphi$,  and any pointed $\funT$-model $(\bbS,s)$:
\[\bbS_*,s_* \sat \varphi \text{ iff } \bbS,s \sat t(\varphi).\]
The proof of this claim was by induction on the complexity of the formula $\varphi$, the crucial step being the case for an existentially quantified formula $\exists p.\alpha$, where the  translation is defined by $t(\exists p.\alpha) := \widetilde{\exists} p. t(\alpha)$. (Here $\widetilde{\exists} p$ is a bisimulation quantifier, defined as in Section \ref{s:ui}.) The problem with proving the required equivalence for this case lies in the direction from left to right. One could try to reason as follows: if $\bbS_*,s_* \sat \exists p.\alpha$, then the model $\bbS_*$ can be extended with a value $Z$ for $p$ such that we have $\bbS_*[p\mapsto Z],s_* \sat \alpha$, where $\bbS_*[p\mapsto Z]$ is the extended model. If we could then conclude that also $\bbS_*[p\mapsto Z],s_* \sat t(\alpha)$, then we would have $\bbS,s \sat \widetilde{\exists} p. t(\varphi)$ by the semantics of the bisimulation quantifier, and the argument would be done. The missing step needed to prove this is to show that the model $(\bbS_*[p \mapsto Z],s_*)$ is a pre-image, i.e. to find a model $(\bbS',s')$ such that $(\bbS_*[p\mapsto Z],s_*) = (\bbS'_*,s'_*)$; we could then argue that $\bbS',s' \sat t(\alpha)$ (by the induction hypothesis), hence $\bbS_*[p\mapsto Z],s_* \sat t(\alpha)$ (since $t(\alpha)$ is a formula of $ \muML_\Lambda(\Prop)$ and therefore bisimulation invariant). This looks plausible at first sight, since the model $\bbS_*[p\mapsto Z]$ is ``almost'' a pre-image of $\bbS$,  aside from the added value for the variable $p$. Unfortunately, the model $\bbS_*[p\mapsto Z]$ does not have to be a pre-image of any model, in general. Without going into a detailed counter-example, the reason is that typically the pre-image construction will create many bisimilar copies of states in a model, where each copy will satisfy the same propositional variables. So even though the model $\bbS_*$ is certainly a pre-image,  the value of $p$ may prevent the extended model $\bbS_*[p\mapsto Z]$ from being a pre-image -- it might not contain enough copies of the states satisfying $p$.

\paragraph{\textbf{Conference version:}}
This article is an extended version of a paper~\cite{enqv:disj17} presented
in Ljubljana at the 2017 conference on Algebra and Coalgebra in Computer
Science (CALCO 2017).
Besides the material presented there, we have provided more detailed proofs,
new examples and the new
results on modular construction of logics admitting disjunctive bases via
products, sums and compositions.

%%% Local Variables:
%%% mode: latex
%%% TeX-master: "main.tex"
%%% End:

%% file: sec-preliminaries.tex
\section{Preliminaries}%
\label{s:prel}

\subsection{Basics of coalgebraic logic}

We assume that the reader is familiar with coalgebra, coalgebraic modal logic
and the basic theory of automata operating on infinite objects.
The aim of this section is merely to fix some definitions and notations on notions
related to coalgebraic modal logic.
In an appendix to this paper we provide some basic definitions related to the
theory of infinite parity games --- we shall need such games for our results on
coalgebraic modal \emph{fixpoint} logics.

First of all, throughout this paper we will use the letter $\funT$ to denote an
arbitrary \emph{set functor}, that is, a covariant endofunctor on the category
$\Set$ having sets as objects and functions as arrows.
For notational convenience we sometimes assume that $\funT$ preserves
inclusions; our arguments can easily be adapted to the more general case.
Functors of coalgebraic interest include the identity functor $\Id$, the
\emph{powerset functor} $\funP$, the monotone neighborhood functor $\funM$
and the (finitary) bag functor $\funB$ (where $\funB S$ is the collection of
\emph{weight functions} $\si: S \to \om$ with finite support).
We also need the contravariant powerset functor $\funQ$.
\begin{definition}
A \emph{$\funT$-coalgebra} is a pair $\bbS = (S,\si)$ where $S$ is a set of
objects called \emph{states} or \emph{points} and $\si : S \to \funT S$ is the
\emph{transition} or \emph{coalgebra map} of $\bbS$.
A \emph{pointed} $\funT$-coalgebra is a pair $(\bbS,s)$ consisting of a
$\funT$-coalgebra and a state $s \in S$.
We call a function $f: S' \to S$ a \emph{coalgebra homomorphism} from $(S',\si')$
to $(S,\si)$ if $\si\circ f = \funT f \circ \si'$, and write $(\bbS',s') \simu
(\bbS,s)$ if there is such a coalgebra morphism mapping $s'$ to $s$.
\end{definition}

Throughout the paper we fix an (unnamed) countable supply of propositional
letters (or variables), of which we often single out a finite subset $\Prop$.

\begin{definition}
With $\Prop$ a set of proposition letters,
a \emph{$\funT$-model over $\Prop$} is a pair $(\bbS,V)$ consisting of a
$\funT$-coalgebra $\bbS = (S,\si)$ and a \emph{$\Prop$-valuation} $V$ on $S$,
that is, a function $V: \Prop \to \funP S$.
The \emph{marking} associated with $V$ is the transpose map $\fl{V}: S \to \funP
\Prop$ given by $\fl{V}(s) \isdef \{ p \in \Prop \mid s \in V(p) \}$.
Thus the pair $(\bbS,V)$ induces a $\funT_{\Prop}$-coalgebra $(S,(\fl{V},\si))$,
where $\funTX$ is the set functor $\funP\Prop \times \funT$.

Given $\funT$-models $\bbS,\bbS'$, a map $f : \bbS \to \bbS'$ is called a
\emph{$\funT$-model homomorphism} if it is a $\funT_\Prop$-coalgebra homomorphism
for the induced $\funT_\Prop$-coalgebras, i.e., it is a $\funT$-coalgebra
homomorphism that preserves the truth values of all propositional variables.
Pointed $\funT$-models $(\bbS,s)$ and $(\bbS',s')$ are said to be
\emph{behaviorally equivalent}, written $(\bbS,s) \beheq (\bbS',s')$, if there is
a pointed $\funT$-model $(\bbS'',s'')$ and $\funT$-model homomorphisms $f :
\bbS \to \bbS''$ and $f' : \bbS' \to \bbS''$ such that $f(s) = f'(s')$.
\end{definition}

We will mainly follow the approach in coalgebraic modal logic where modalities
are associated (or even identified) with finitary predicate liftings.
\begin{definition}
A \emph{predicate lifting} of arity $n$ is a natural transformation $\la:
\funQ^{n} \Rightarrow \funQ\funT$.
Such a predicate lifting is \emph{monotone} if for every set $S$, the map
$\la_{S}: {(\funQ S)}^{n} \to \funQ\funT S$
% is order-preserving in each coordinate (with respect to the subset order).
preserves the subset order in each coordinate.
The induced predicate lifting $\bdual{\la}: \funQ^{n} \Rightarrow \funQ\funT$,
given by $\bdual{\la}_{S}(X_{1},\ldots,X_{n}) \isdef
\funT S \;\setminus\; \la_{S}(S\setminus X_{1}, \ldots, S\setminus X_{1})$,
is called the \emph{(Boolean) dual} of $\la$.
\end{definition}
\begin{definition}
A \emph{monotone modal signature}, or briefly: \emph{signature} for $\funT$ is
a set $\La$ of monotone predicate liftings for $\funT$, which is closed under
taking boolean duals.
\end{definition}

In this paper we will study coalgebraic modal logic with and without fixpoint
operators.
Given a signature $\La$, the formulas of the \emph{coalgebraic $\mu$-calculus}
$\muML_{\La}$ are given by the following grammar:
\[
\phi \isbnf p
   \divbnf \bot
   \divbnf \neg \phi
   \divbnf \phi_{0} \lor \phi_{1}
   \divbnf \hs_{\la}(\phi_{1},\ldots,\phi_{n})
   \divbnf \mu q.\phi'
\]
where $p$ and $q$ are propositional variables, $\la \in \La$ has arity $n$, and
the application of the fixpoint operator $\mu q$ is under the proviso that all
occurrences of $q$ in $\phi'$ are positive (i.e., under an even number of
negations).
We let $\MLLa$ denote the fixpoint-free fragment of $\muMLLa$, i.e., the ``basic''
coalgebraic modal logic of the signature $\La$.   We let  $\muML_{\La}(\Prop)$
denote  the set of $\muML_{\La}$-formulas taking free variables from $\Prop$,
and define the notation $\MLLa(\Prop)$ similarly.

% Since we only consider signatures that are closed under taking boolean duals,
% we may rewrite every formula into \emph{negation normal form}, i.e. a formula
% in the grammar
% \[
% \phi \isbnf p \divbnf \neg p
%    \divbnf \bot \divbnf \top
%    \divbnf \phi_{0} \lor \phi_{1} \divbnf \phi_{0} \land \phi_{1}
%    \divbnf \hs_{\la}(\phi_{1},\ldots,\phi_{n})
%    \divbnf \mu x.\phi' \divbnf \nu x.\phi'
% \]

Formulas of such coalgebraic $\mu$-calculi are interpreted in coalgebraic
models, as follows.
Let $\bbS = (S,\si,V)$ be a $\funT$-model over a set $\Prop$ of proposition
letters.
By induction on the complexity of formulas, we define a \emph{meaning function}
$\mng{\cdot}^{\bbS}: \muML_{\La}(\Prop) \to \funP S$, together with an associated
\emph{satisfaction relation} ${\sat} \sse S \times \muML_{\La}(\Prop)$ given
by $\bbS,s \sat \phi$ iff $s \in \mng{\phi}^{\bbS}$.
All clauses of this definition are standard; for instance, the one for the
modality $\hs_{\la}$ is given by
\begin{equation}
\label{eq:semhs}
\bbS,s \sat \hs_{\la}(\phi_{1},\ldots,\phi_{n}) \text{ if }
\si(s) \in \la_{S}(\mng{\phi_{1}}^{\bbS},\ldots,\mng{\phi_{n}}^{\bbS}).
\end{equation}
For the least fixpoint operator we apply the standard description of least
fixpoints of monotone maps from the Knaster-Tarski theorem and take
\[
\mng{\mu x. \phi}^{\bbS} \isdef
  \bigcap \big\{ U \in \funP S \mid \mng{\phi}^{(S,\si,V[x \mapsto U])} \sse U
   \big\},
\]
where $V[x \mapsto U]$ is given by $V[x \mapsto U](x) \isdef U$
while $V[x \mapsto U](p) \isdef V(p)$ for $p \neq x$.
A formula $\phi$ is said to be \emph{monotone} in a variable $p$ if, for every
$\funT$-model $\bbS = (S,\sigma,V)$ and all sets $Z_1 \subseteq Z_2 \subseteq S$,
we have $ \mng{\phi}^{(S,\si,V[p \mapsto Z_1])} \subseteq
\mng{\phi}^{(S,\si,V[p \mapsto Z_2])}$.

\subsection{Examples}%
\label{s:someexamples}

There are many well-known examples of modal logics that can be presented as
coalgebraic modal logics where the modalities correspond to predicate liftings
for the relevant functor. We shall not attempt to provide a complete list here,
but we provide a few basic examples that will be helpful in what follows.

\paragraph*{Next-time modality}
Probably the simplest example of a non-trivial modality that can be described as
a predicate lifting is the ``next-time'' operator of linear temporal logic.
The natural way to present models of $\mathtt{LTL}$ coalgebraically is to take
the coalgebraic type functor to be $\funI$, the identity functor on the category
of sets.
Coalgebras $(S,\sigma)$ for this functor just provide maps $\sigma : S \to S$,
which can be thought of as providing the ``next-state'' function for a discrete
linear flow of time.

The identity functor has a unary predicate lifting $\nxt{} : \funQ\to \funQ$:
the identity natural transformation defined by $\nxt{S} : Z \mapsto Z$ for
$Z \subseteq S$. A little thought shows that the evaluation of formulas
$\nxt{}\varphi$ in a model turns out as expected:
\[
    \bbS,s \sat \nxt{}\varphi \text{ iff } \bbS,\sigma(s) \sat \varphi.
\]
The lifting $\nxt{}$ is monotone and dual to itself, so $\Si_{\funI} =
\{\nxt{}\}$ is a modal signature for the identity functor.
The language $\muML_{\La}$ induced by this signature is known as the
\emph{linear-time $\mu$-calculus}.

\paragraph*{Basic modal logic}
Kripke frames for modal logic are coalgebras for the covariant powerset functor
$\funP$, or rather can be represented equivalently as such: a binary relation
$R \subseteq S \times S$ of a frame $(S,R)$ can be identified with the map
$R[-] : S \to \funP S$ sending each point in $S$ to its set of $R$-successors.
The usual Kripkean modalities $\Box$ and $\dia$ come out as predicate liftings
for $\funP$, by setting $\Box_S(Z) = \{Z' \subseteq S \mid Z' \subseteq Z\}$
and $\Diamond_S(Z) = \{Z' \subseteq S \mid Z \cap Z' \neq \nada\}$.
Unfolding the definitions we see that the naturality condition for the box
modality $\Box$ says that for all $f : S \to S'$ and $Z \subseteq S$,
$Z' \subseteq S'$, we have:
\[
f[Z] \subseteq Z'\; \Longleftrightarrow \; Z \subseteq f^{-1}[Z'].
\]
This is the familiar adjunction between direct and inverse image.
The semantics of modal formulas comes out as expected: given a Kripke model
$(S,R,V)$ represented as a $\funP$-model $\bbS$, we have $\bbS,s \sat \Box
\varphi$ if and only if $\bbS,t \sat \varphi$ for every $R$-successor $t$ of
$s$.
For the signature $\Si_{\funP} = \{\Box,\dia\}$, the language
$\muML_{\Si_{\funP}}$ is the (standard) modal $\mu$-calculus.

\paragraph*{Monotone modal logic}

Monotone modal logic generalizes normal modal logic by dropping the constraint
that the box modality should commute with conjunctions: $\Box (\varphi \wedge
\psi) \Leftrightarrow \Box \varphi \wedge \Box \psi$.
Semantics for monotone modal logic is given by coalgebras for the monotone
neighborhood functor $\funM$, which is the sub-functor of the functor $\funQ
\circ \funQ$ defined by setting:
\[
\funM S = \{F \subseteq \funQ S \mid Z \in F
\; \& \; Z \subseteq Z' \; \Rightarrow \; Z' \in F\}
\]
The lifting $\Box : \funQ \to \funQ \circ \funM$ is defined by setting $F \in
\Box_S(Z)$ iff $Z \in F$.
The dual lifting $\Diamond$ is defined by setting $F \in \Diamond_S(Z)$ iff,
for all $Z'\in F$, $Z \cap Z' \neq \nada$.

Monotone modal logic can be seen as a ``base logic'' for modal logics without
distribution of the box over conjunctions, with varying interpretations, much
as the modal logic $\mathbf{K}$ can be seen as the most basic normal modal logic.
Examples where such modalities appear are alternating-time temporal
logic~\cite{alur:alte02} and Parikh's dynamic game logic~\cite{pari:logi85}.
For  $\Si_{\funM} = \{\Box,\dia\}$, the language $\muML_{\Si_{\funM}}$ is known
as the \emph{monotone $\mu$-calculus}.
It stands in a similar relationship to alternating-time logic and game logic as
the modal $\mu$-calculus does to $\mathtt{CTL}$ and $\mathtt{PDL}$.

\paragraph*{Graded modal logic}
Graded modal logic extends basic modal logic with \emph{counting modalities}
$\dia^{k} \varphi$ and $\Box^{k}\varphi$.
These modalities, interpreted on Kripke models, are interpreted as ``at least
$k$ successors satisfy $\varphi$'' and ``there are less than $k$ successors that
do not satisfy $\varphi$''.
It is often convenient to use a slight generalization of Kripke semantics, where
successors of a state are assigned ``weights'' from $\omega$.
Such models are based on  $\funB$-coalgebras, where the ``bags'' functor $\funB$
assigns to a set $S$ the set $\funB S$ of maps $f : S \to \omega$ such that
$f(s) = 0$ for all but finitely many $s \in S$.
Given a map $h : S \to S'$ and $f \in \funB S$, the map $f' = \funB h(f)$ is
defined by:
\[
    f'(s') = \sum_{h(s) = s'} f(s)
\]
The functor $\funB$ comes with an infinite supply of predicate liftings
$\underline{k}$ and $\ol{k}$ --- one pair for each $k \in \omega$ --- given by:
\[
\begin{array}{lll}
\underline{k}_{S}: &
      U \mapsto \big\{ \si \in \funB S \mid \sum_{u \in U} \si(u) \geq k \big\}
\\[1mm]
\ol{k}_{S}: &
      U \mapsto \{ \si \in \funB S \mid \sum_{u \not\in U} \si(u) < k  \big\}.
\end{array}
\]
Over Kripke models, which can be identified with $\funB$-models $(S,\sigma,V)$
in which $\sigma(s)(v) \in \{0,1\}$ for all $s,v \in S$, it is not hard to see
that the formulas $\dia^{k} \varphi$ and $\Box^{k}\varphi$ get their expected
meanings.
For $\Si_{\funB} = \{\underline{k} \mid k \in \omega\} \cup \{\ol{k} \mid k \in
\omega\}$, the language $\muML_{\Si_{\funB}}$ is known as the \emph{graded
$\mu$-calculus}.

\subsection{One-step logic and one-step models}

A pivotal role in our approach is filled by the \emph{one-step versions} of coalgebraic
logics.  The one-step perspective on coalgebraic modal logics, developed by a number of authors over several papers~\cite{cirs:modu04,patt:coal03,schr:pspa09,schr:rank10}, has been key to proving some central results about such logics and their fixpoint extensions. In particular, it is instrumental  to the theory of coalgebraic automata. Indeed our main contribution here --- the concept of disjunctive bases --- takes place on the one-step level.
We do not assume that the reader is familiar with the framework of one-step logics, and give a self-contained introduction here.

We begin with a formal definition.
\begin{definition}
Given a signature $\Lambda$ and a set $A$ of variables, we define the set
$\Bool(A)$ of \emph{boolean formulas} over $A$ and the set $\MLoLa(A)$
of \emph{one-step $\La$-formulas} over $A$, by the following grammars:
\begin{eqnarray*}
\Bool(A) \ni \pi & \isbnf &
   a \divbnf
   \bot \divbnf \top \divbnf
   \pi \vee  \pi \divbnf \pi  \wedge  \pi
   \divbnf \neg \pi
\\ \MLoLa (A) \ni \al & \isbnf &
   \hs_\la \ol{\pi} \divbnf
   \bot \divbnf \top \divbnf
   \al \vee \al \divbnf \al \wedge \al
   \divbnf \neg \al
\end{eqnarray*}
where $a \in A$, $\la \in \Lambda$ and $\ol{\pi} = (\pi_{1},\ldots,\pi_{n})$ is
a tuple of formulas in $\Bool(A)$ of the same length as the arity of $\la$.
We will denote the positive (negation-free) fragments of $\Bool(A)$ and
$\MLoLa(A)$ as, respectively, $\Latt(A)$ and $\MLoLa^{+}(A)$:
\begin{eqnarray*}
\Latt(A) \ni \pi & \isbnf &
   a \divbnf
   \bot \divbnf \top \divbnf
   \pi \vee  \pi \divbnf \pi  \wedge  \pi
\\ \MLoLa^{+} (A) \ni \al & \isbnf &
   \hs_\la \ol{\pi} \divbnf
   \bot \divbnf \top \divbnf
   \al \vee \al \divbnf \al \wedge \al
\end{eqnarray*}

We shall often make use of substitutions: given a finite set $A$,
let $\vee_A : \funP A \to \Latt(A)$ be the map sending $B$ to $\bigvee B$, and
let $\wedge_A : \funP A \to \Latt(A)$ be the map sending $B$ to $\bigwedge B$,
and given sets $A,B$ let $\theta_{A,B} : A \times B \to \Bool(A \cup B)$ be
defined by mapping $(a,b)$ to $a \wedge b$.
\end{definition}

We need a number of properties of modal signatures, formulated in terms of the corresponding one-step logics. The first is fairly standard: expressive completeness of a modal signature $\Lambda$ means that every modality that makes sense for the functor --- formally, every predicate lifting --- can be expressed in terms of modalities in $\Lambda$.

\begin{definition}
A monotone modal signature $\La$ for $\funT$ is \emph{expressively complete} if,
for every  $n$-place predicate lifting $\la$ (not necessarily in $\La$) and for
all variables $a_{1},\ldots, a_{n}$ there is a formula $\al\in
\MLoLa(\{a_{1},\ldots, a_{n}\})$ which is equivalent to $\hs_{\la}\ol{a}$.
\end{definition}
We will also be interested in the following variant of expressive completeness:
\begin{definition}
We say that $\Lambda$ is \emph{Lyndon complete} if, for every \emph{monotone}
$n$-place predicate lifting $\la$ and variables $a_{1},\ldots, a_{n}$, there is
a \emph{positive} formula $\al\in  \MLoLa^+(\{a_{1},\ldots, a_{n}\})$ equivalent
to $\hs_{\la}\ol{a}$.
\end{definition}
% \begin{definition}
% A signature $\La$ is said to be \emph{expressively complete} if, for every
% $n$-place predicate lifting $\lambda$ for $\funT$, there is a formula
% $\alpha \in \MLoneLa(\{1,...,n\})$ such that, for all one-step frames $(X,\xi)$
% and $Z_1,...,Z_n \subseteq X$, we have $\xi \in \lambda_X(Z_1,...,Z_n)$ iff
% $(X,\xi,m) \satone \alpha$, where $m : X \to \{1,...,n\}$ is defined by
% $m(u) = \{i \mid u \in Z_i\}$.
% \end{definition}

We now turn to the semantics of one-step formulas, using so-called \emph{one-step models}.
\begin{definition}
A \emph{one-step $\funT$-frame} is a pair $(S,\si)$ with $\si \in \funT S$,
i.e., an object in the category $\mathcal{E}(\funT)$ of elements of $\funT$.
%which we denote by $\mathcal{E}(\funT)$.
Similarly a \emph{one-step $\funT$-model} over a set $A$ of variables is a
triple $(S,\si,m)$ such that $(S,\si)$ is a one-step $\funT$-frame and $m:
S \to \funP A$ is an $A$-marking on $S$.

A \emph{morphism} $f : (S,\si) \to (S',\si')$ is a morphism in
$\mathcal{E}(\funT)$, that is, a map from $S$ to $S'$ such that $\funT f(\si)
= \si'$.
The map $f$ is said to be a \emph{morphism of one-step models} $f : (S,\si,m)
\to (S',\si',m')$ if, in addition, $m = m' \circ f$.
\end{definition}

Given a one-step model $(S,\si,m)$, we define the \emph{$0$-step interpretation}
$\mngzero{\pi}_{m} \sse S$ of $\pi \in \Bool(A)$ by the obvious induction:
$\mngzero{a}_m \isdef \{v \in S \mid a \in m(v)\}$, $\mngzero{\top}_m \isdef S$,
$\mngzero{\bot}_m \isdef \nada$, while we use standard clauses for $\wedge,
\vee$ and $\neg$.
Similarly, the \emph{one-step interpretation} $\mngone{\al}_m$ of $\al \in
\MLoLa(A)$ is defined as a subset of $\funT S$, with
$\mngone{\hs_\la (\pi_{1},\ldots,\pi_{n})}_m \isdef
\la_S(\mngzero{\pi_{1}}_m,\ldots,\mngzero{\pi_{n}}_m)$,
and again standard clauses apply to $\bot, \top, \wedge, \vee$ and $\neg$.
Given a one-step model $(S,\si,m)$, we write $S,\si, m \satone \al$ for
$\si \in \mngone{\al}_m$.
Notions like one-step satisfiability, validity and equivalence are defined and
denoted in the obvious way; in particular, we use $\equiv^{1}$ to denote the
equivalence of one-step formulas.

For future reference we mention the following two results, the first of which
states that the truth of one-step formulas is invariant under one-step
morphisms.

\begin{prop}%
\label{p:1invar}
Let $f: (S',\si',m') \to (S,\si,m)$ be a morphism of one-step models over $A$.
Then for every formula $\al \in \MLoLa(A)$ we have
\[
S',\si',m' \satone \al \text{ iff } S,\si,m \satone \al.
\]
\end{prop}

The second proposition is a standard observation about the semantic counterpart
of the syntactic notion of substitution.

\begin{prop}%
\label{p:1subst}
Let $(S,\si,m)$ be a one-step model over $A$, and let $\si: B \to \Bool(A)$ be
a substitution.
Then for every formula $\al \in \MLoLa(B)$ we have
\[
S,\si,m_{\si} \satone \al \text{ iff } S,\si,m \satone \al[\si],
\]
where $m_{\si}$ is the $B$-marking given by $m_{\si}(b) \isdef
\mngzero{\si_{b}}_{m}$.
\end{prop}

\subsection{Graph games}%
\label{sec:games}

For readers unfamiliar with the theory of infinite games, we provide some of
the basic definitions here, referring to~\cite{grae:auto02} for a survey.

\begin{definition}%
\label{d:game}
A \emph{board game} is a tuple $\bbG = (G_{\eloi},G_{\abel},E,W)$ where
$G_{\eloi}$ and $G_{\abel}$ are disjoint sets, and, with $G \isdef G_{\eloi}
\cup G_{\abel}$ denoting the \emph{board} of the game, the binary relation
$E \subseteq G^2$ encodes the moves that are admissible to the respective
players, and $W\subseteq G^\omega$ denotes the \emph{winning condition}
of the game.
In a \emph{parity game}, the winning condition is determined by a parity map
$\Om: G \to \om$ with finite range, in the sense that the set
$W_{\Om}$ is given as the set of $G$-streams $\rho \in G^{\om}$
such that the maximum value occurring infinitely often in the stream
${(\Om\rho_{i})}_{i\in\om}$ is even.

Elements of $G_{\eloi}$ and $G_{\abel}$ are called \emph{positions} for the
players $\eloi$ and $\abel$, respectively; given a position $p$ for player
$\Pi \in \{ \eloi, \abel\}$, the set $E[p]$ denotes the set of \emph{moves}
that are \emph{legitimate} or \emph{admissible to} $\Pi$ at $p$.
In case $E[p] = \nada$ we say that player $\Pi$ \emph{gets stuck} at $p$.

An \emph{initialized board game} is a pair consisting of a board game $\bbG$
and a \emph{initial} position $p$, usually denoted as $\bbG@p$.
\end{definition}

\begin{definition}%
\label{d:match}
A \emph{match} of a graph game $\bbG = (G_{\eloi},G_{\abel},E,W)$ is
nothing but a (finite or infinite) path through the graph $(G,E)$.
Such a match $\rho$ is called \emph{partial} if it is finite and $E[\last\rho]
\neq\nada$, and \emph{full} otherwise.
We let $\PM{\Pi}$ denote the collection of partial matches $\rho$ ending in a
position $\last(\rho) \in G_{\Pi}$, and define $\PM{\Pi}@p$ as the set of
partial matches in $\PM{\Pi}$ starting at position $p$.

The \emph{winner} of a full match $\rho$ is determined as follows.
If $\rho$ is finite, then by definition one of the two players got stuck at
the position $\last(\rho)$, and so this player looses $\rho$, while the opponent
wins.
If $\rho$ is infinite, we declare its winner to be $\eloi$ if $\rho
\in W$, and $\abel$ otherwise.
\end{definition}

\begin{definition}
A \emph{strategy} for a player $\Pi \in \{ \eloi,\abel \}$ is a map $\chi:
\PM{\Pi} \to G$.
A strategy is \emph{positional} if it only depends on the last position of a
partial match, i.e., if $\chi(\rho) = \chi(\rho')$  whenever $\last(\rho) =
\last(\rho')$; such a strategy can and will be presented as a map $\chi:
G_{\Pi} \to G$.

A match $\rho = {(p_{i})}_{i<\kappa}$ is \emph{guided} by a $\Pi$-strategy
$\chi$ if $\chi(p_{0}p_{1}\ldots p_{n-1}) = p_{n}$ for all $n<\kappa$
such that $p_{0}\ldots p_{n-1}\in \PM{\Pi}$ (that is, $p_{n-1}\in G_{\Pi}$).
Given a strategy $f$, we say that a position $p$ is $f$-reachable if $p$
occurs on some $f$-guided partial match.
A $\Pi$-strategy $\chi$ is \emph{legitimate} in $\bbG@p$ if the moves that it
prescribes to $\chi$-guided partial matches in $\PM{\Pi}@p$ are always
admissible to $\Pi$, and \emph{winning for $\Pi$} in $\bbG@p$ if in addition
all $\chi$-guided full matches starting at $p$ are won by $\Pi$.

A position $p$ is a \emph{winning position} for player $\Pi \in \{ \eloi, \abel
\}$ if $\Pi$ has a winning strategy in the game $\bbG@p$; the set of these
positions is denoted as $\Win_{\Pi}$.
The game $\bbG = (G_{\eloi},G_{\abel},E,W)$ is \emph{determined} if every
position is winning for either $\eloi$ or $\abel$.
\end{definition}

When defining a strategy $\chi$ for one of the players in a board game, we
will often confine ourselves to defining $\chi$ for partial matches
that are themselves guided by $\chi$.
The following fact, independently due to Emerson \& Jutla~\cite{emer:tree91}
and Mostowski~\cite{most:game91}, will be quite useful to us.

\begin{myfact}[Positional Determinacy]%
\label{f:pdpg}
Let $\bbG = (G_{\eloi},G_{\abel},E,W)$ be a graph game.
If $W$ is given by a parity condition, then $\bbG$ is determined, and both
players have positional winning strategies.
\end{myfact}

\subsection{Automata}%
\label{ss:aut}

Given a state $s$ in a coalgebra $\sigma : S \to \funT S$, consider
the one-step frame $(S,\sigma(s))$ --- this is the local ``window'' into the
structure of the coalgebra that is directly visible from $s$.
The main function of one-step models  for our purposes here is to provide a
neat framework for \emph{automata} running on $\funT$-models.
The idea is the following: at any stage in the run of an automaton $\bbA$ on
a model $\bbS = (S,\si,V)$, the automaton reads some point $s$ in $\bbS$ and
takes the set $V^{\flat}(s) \sse \Prop$, consisting of those propositional
variables that are true at $s$, as input.
Next the automaton decides whether to continue the computation or to reject.
To decide this, the automaton checks the directly visible part of the model
$\bbS$ at $s$, modelled as the one-step frame $\sigma(s)$, and looks for an
admissible way to continue the computation one step further.
The ways in which the run may continue are constrained by a one-step formula
$\alpha$, which depends on the current state $a$ of the automaton and the last
input $V^{\flat}(s)$ read, and is built up using states of the automaton as
variables.
The run continues if a marking $m$ can be found that makes $\alpha$ true in
the one-step model $(S,\sigma(s),m)$, and the marking $m$ then determines which
states in $\bbS$ may be visited in the next stage of the run, and the possible
next states of the automaton.

\begin{definition}
A \emph{$(\La,\Prop)$-automaton}, or more broadly, a \emph{coalgebra automaton},
is a quadruple $(A,\Th,\Om,\ai)$ where $A$ is a finite set of \emph{states},
with \emph{initial state} $\ai \in A$,
$\Th : A \times \funP\Prop \to \MLoLa^{+}(A)$ is the \emph{transition map}
and $\Om : A \to \omega$ is the \emph{priority map} of $\bbA$.
\end{definition}
% In case $\Prop$ is clear from context or irrelevant, we may speak of
% $\La$-automata, and if we discuss automata for an arbitrary signature we will
% use the term \emph{coalgebra automata}.
%
The  semantics of such an automaton is given in terms of a two-player infinite parity game:
% Let $\bbA = (A,\Th,\Om,\ai)$ be a $(\La,\Prop)$-automaton, and let
With $\bbS = (S,\si,V)$ a $\funT$-model over a set $\PropQ \supseteq \Prop$,
the \emph{acceptance game} $\AG(\bbA,\bbS)$ is the parity game given by the
table below.

\begin{center}
\begin{tabular}{|l|c|l|c|} % chktex 44
\hline % chktex 44
    Position             & Player
    &  Admissible moves  & Priority
\\ \hline % chktex 44
   $(a,s)\in A\times S$
   & $\eloi$
   & $\{m : S \to \funP A \mid (S,\si(s),m) \satone \Th(a,\Prop\cap V^{\flat}(s)) $\}
   & $\Om(a)$
\\ $m: S \to \funP A$
   & $\abel$
   & $\{(b,t)\mid b \in m(t) \}$
   & 0
\\ \hline % chktex 44
\end{tabular}
\end{center}

We say that $\bbA$ \emph{accepts} the pointed $\funT$-model $(\bbS,s)$,
notation: $\bbS,s \sat \bbA$, if $(\ai,s)$ is a winning position for $\eloi$
in the acceptance game $\AG(\bbA,\bbS)$.

The connection with coalgebraic modal logic is given by the following
result.

\begin{myfactC}[\cite{font:auto10}]
There are effective constructions transforming a formula in $\muMLLa(\Prop)$
into an equivalent $(\La,\Prop)$-automaton, and vice versa.
\end{myfactC}

% For future reference we give the following fact.
% \begin{myfact}[Invariance]
% \label{f:inv}
% Let $\bbA$ be a $\La$-automaton, then for any two $\funT$-models $(\bbS,s)$ and
% $(\bbS',s')$ such that $(\bbS',s') \simu (\bbS,s)$ we have that
% $ \bbS,s \sat \bbA \text{ iff } \bbS',s' \sat \bbA$.
% \end{myfact}

%%% Local Variables:
%%% mode: latex
%%% TeX-master: "main.tex"
%%% End:

%% file: sec-dis-bases.tex
\section{Disjunctive formulas and disjunctive bases}%
\label{s:dis}
\subsection{Disjunctive formulas}
In this section, we present the main conceptual contribution of the paper, and
define disjunctive bases. We then consider a number of examples.

As a first step, we begin by presenting disjunctive \emph{formulas}, originally
introduced in~\cite{enqv:comp17}, as a class of one-step formulas for a given
modal signature characterized by a model-theoretic property, expressed in terms
of the one-step semantics.

\begin{definition}%
\label{d:dj}%
\label{d:disfma}
A one-step formula $\al \in \MLoLa^{+}(A)$ is called \emph{disjunctive}
if for every one-step model $(S,\si,m)$ such that $S,\si,m \satone \al$ there is
a one-step frame  $ (S',\si')$ together with a one-step frame morphism $f : (S',\si') \to (S,\si)$ and a marking $m': S'\to
\funP A$, such that:
\begin{enumerate}
    \item $S',\si',m' \satone \al$;
    \item $m'(s') \sse m(f(s'))$, for all $s' \in S'$;
    \item $\sz{m'(s')} \leq 1$, for all $s' \in S'$.
\end{enumerate}
We sometimes refer to the one-step frame $(S',\si')$ together with the map $f$
as a \emph{cover} of $(S,\si)$, and to the one-step model $(S',\si',m')$ together
with the  map $f$ as a \emph{dividing cover} of $(S,\si,m)$ for $\al$.
\end{definition}

The intuition behind disjunctive formulas is that, in a certain sense, they
never ``force'' two distinct propositional variables to be true together, i.e.
any one-step model for a disjunctive formula $\delta$ in $\MLoLa^{+}(A)$ can be
transformed into one in which every point satisfies at most one propositional
variable from $A$.
Moreover, ``transformed into'' here does not just mean ``replaced by'': we
cannot arbitrarily change the one-step model, the output of the construction
must be closely related to the one-step model that we started with.

A trivial example of a disjunctive formula is $\nxt{} a$ for $a \in A$, where
we recall that $\nxt{}$ was the next-time modality viewed as a predicate
lifting for the identity functor $\mathsf{Id}$.
A one-step model for this functor is a triple $S,s,m$ consisting of a set $S$,
an element $s \in S$ and a marking $m : S \to \funP A$. Then $S,s,m \satone
\nxt{} a$ if, and only if, $a \in m(s)$.
But then, no elements in $S$ besides $s$ are relevant to the evaluation of
$\nxt{}a$, and for $s$ we can just forget about all other variables: set
$m'(s) =  \{a\}$ and $m'(v) = \nada$ for all $s \in S\setminus\{v\}$.
We have $S,s,m' \satone \nxt{} a$, $m'(v) \subseteq m(v)$ and $\sz{m'(v)}
\leq 1$ for all $v \in S$.

For an example of a one-step formula that is \emph{not} disjunctive, consider
$\Box a \land \dia b$ (where $\dia$ and $\Box$ are the standard modalities for
Kripke structures, i.e., coalgebras for the power set functor $\funP$).
Observe that a one-step model for this functor is a triple $(S,\si,m)$ with
$\si \sse S$.
It should be obvious that for the formula $\Box a \land \dia b$ to hold at such
a structure, $\si$ needs to have an element $s$ where $b$ holds, while at the
same time \emph{every} element of $\si$, including $s$, must satisfy $a$.
There is no escape here: we can only have $(S,\si,m) \satone \Box a \land \dia b$
if there is an element $s$ making \emph{both} $a$ and $b$ true.

This is very different if we consider
the typical disjunctive formulas for basic modal logic, which are of the
form:
\[
    \dia a_1 \wedge \dots \wedge \dia a_n \wedge \Box (a_1 \vee \dots \vee a_n)
\]
This is often abbreviated as $\nabla \{a_1,\dots,a_n\}$.
The operator $\nabla$ is known as the \emph{cover modality}, or sometimes
``nabla modality''.
The formula $\nabla B$ says about a one-step model $(S,\si,m)$ for $\funP$,
that the following ``back-and-forth'' conditions hold: for all $s \in \si$
there is some $a \in B$ with $a \in m(s)$, and conversely, for every $a \in B$
there is some $s \in \si$ with $a \in m(s)$.
These formulas are indeed disjunctive, but less trivially so than the next-time
formulas.
For example, consider the formula $\nabla \{a,b,c\}$, which is true in a
one-step model $(S,s,m)$ with $S = \{u,v\}$, $m(u) = \{a,b\}$ and $m(v) =
\{b,c\}$.
But there is no way to simply shrink the marking $m$ to a dividing marking
$m'$ so that $(S,s,m) \satone \nabla \{a,b,c\}$:
there are too many pigeons and too few pigeon holes --- so the obvious solution
is to make more pigeon holes!
One way to do this is to ``split'' the points in $\{ u,v \}$
so that we make room for each variable to be witnessed at a separate point.
More formally, let $S' = \{ (u,a), (u,b), (v,b), (v,c) \}$ and define  $m' : S'
\to \{a,b,c\}$ via the projection on $\{ a,b,c \}$, i.e., $m'(u,a) = \{a\}$,
etc.
This one-step model satisfies $\nabla \{a,b,c\}$ and it has the obvious covering
map $h$ being the projection $\pi_{S}$ on $S$.

\begin{figure}
\begin{center}
\[
\xymatrixrowsep{0.5pt}\xymatrix{
   \{a\} & (u,a)  \ar@{.>}^{\pi_{S}}[rrd]& & &
\\ & & & u & \{a,b\} \\
   \{b\}& (u,b) \ar@{.>}_{\pi_{S}}[rru]& & &
\\ \{b\} & (v,b)  \ar@{.>}^{\pi_{S}}[rrd]& & &
\\ & & & v & \{b,c\} \\
 \{c\}& (v,c) \ar@{.>}_{\pi_{S}}[rru]& & &
}
\]
\caption{The cover $h : (S',m') \to (S,m)$}\label{f:simplecover}
\end{center}
\end{figure}

The cover modality was arguably the starting point of coalgebraic modal logic.
In the seminal paper~\cite{moss:coal99}, Moss defined nabla modalities for all
functors that preserve weak pullback squares, generalizing the cover modality.
The idea is to apply the coalgebraic type functor $\funT$ to sets of formulas
$\Psi$ and form new formulas from objects $\Gamma$ in $\funT \Psi$.
That is,  $\nabla \Gamma$ counts as a formula if $\Psi$ is a set of formulas
and $\Gamma \in \funT \Psi$.
The semantics is given in terms of the ``Barr extension'' $\overline{\funT}$ of
the functor $\funT$, which is a relation lifting defined for $R \subseteq X
\times A$ by setting:
\[
\overline{\funT}R = \{(\xi,\al) \in \funT X \times \funT A \mid
  \exists \rho \in \funT R: \; \funT\pi_X(\rho) = \xi \; \& \;
  \funT\pi_A(\rho) = \al \}
\]
We then evaluate the nabla modality by applying this relation lifting to the
satisfaction relation:
\[
\bbS,s \sat \nabla \Gamma \text{ iff }
(\sigma(s),\Gamma) \in {\overline{\funT}({\sat})}.
\]

In terms of one-step formulas, we would count $\nabla \Gamma$ as a one-step
formula with variables in $A$, for all $\Gamma \in \funT \Bool(A)$.
In particular, formulas of the form $\nabla \Gamma$ for $\Gamma \in \funT A$
count as one-step formulas.

These formulas are in fact disjunctive:
\begin{prop}
Let $\funT$ be a set functor that preserves weak pullbacks.
Then every formula of the form $\nb\Ga$, where $\Ga \in \funT A$, is disjunctive.
\end{prop}

To see this, suppose that a one-step model $X,\xi,m$ satisfies the formula
$\nabla \Gamma$.
This means that $(\xi,\Gamma) \in \ol{\funT}R$, where $R = {\satzero}
\cap (X \times A)$ is defined by $(u,a) \in R$ iff $a \in m(u)$.
By definition of the Barr extension, this means we find some $\rho \in \funT R$
such that $\funT \pi_A(\rho) = \Gamma$ and $\funT\pi_X (\rho) = \xi$.
But then we have our covering one-step model right in front of us already: the
triple $(R,\rho,\eta \circ \pi_A)$ is a one-step model where $\eta : A \to
\funP A$ is the singleton map $a \mapsto \{a\}$, i.e., the unit of the powerset
monad.
The covering map is just the projection $\pi_X : R \to X$.
It is a simple exercise to check that $(R,\rho,\eta \circ \pi_A) \satone
\nabla \Gamma$, and clearly $\sz{\eta (\pi_A(u,a))} = 1$ and $\eta (\pi_A(u,a))
\subseteq m(u)$ for all $(u,a) \in R$.
So we see that disjunctiveness of the nabla modalities is not incidental: it is
hardwired into their semantics.

\subsection{Disjunctive bases}

We now turn to a discussion of the key notion of this paper, viz., the
disjunctive bases originating with~\cite{enqv:comp17}.
For their definition, we recall the substitutions $\wedge_{A}: B \mapsto \bw B$
and $\theta_{A,B}: (a,b) \mapsto a \land b$ defined in the Preliminaries.
It will also be convenient to introduce the abbrevation
\[
    A \utimes B : = (A \times B) \cup A \cup B
\]
for any two sets $A,B$.

\begin{definition}%
\label{d:disbas}
Let $\D$ be an assignment of a set of positive one-step formulas
$\D(A) \subseteq \MLoLa^{+}(A)$ for all sets of variables $A$.
Then $\D$ is called a \emph{disjunctive basis} for $\La$ if each
formula in $\D(A)$ is disjunctive, and the following conditions
hold:
\begin{enumerate}
    \item $\D(A)$ contains $\top$ and is closed under finite disjunctions (in particular, also $\bot = \bv\nada \in \D(A)$).
    \item $\D$ is \emph{distributive over $\La$}:  for every one-step formula of the form $\hs_{\la}\ol{\pi}$ in $ \MLoLa^{+}(A)$ there is a  formula $\delta \in \D(\psf(A))$ such that $\hs_{\la} \ol{\pi} \oneseq \delta[\land_A]$.
    \item $\D$ \emph{admits a binary distributive law}: for any two formulas $\al \in \D(A)$ and $\beta \in \D(B)$, there is a formula $\gamma \in \D(A \utimes B )$ such that $\al \wedge \beta \oneseq \gamma[\theta_{A,B}]$.
\end{enumerate}
%
%As a stronger version, we say that $\D$ is a \emph{uniform disjunctive basis}
%if each one-step frame  $(S,\si)$ has a single cover $f : (S',\si') \to (S,\si)$
%which works for all $\D$-formulas, in the sense that for each $\delta \in \D(A)$
%and each marking $m : S \to \funP A$ with $S,\si,m \satone \delta$ there is a
%marking $m': S'\to \funP A$ such that $(S',\si',m')$ together with the  map $f$
%provides a dividing cover of $(S,\si,m)$ for $\de$.
\end{definition}

The key property of disjunctive bases is captured by the following normal form
theorem, which is easy to derive from the definition.

\begin{prop}%
\label{p:db}
Suppose $\D$ is a disjunctive basis for $\La$.
Then for every one-step formula $\alpha \in  \MLoLa^{+}(A)$ there is a formula
$\delta \in \D(\funP(A))$ such that $\alpha \oneseq \delta[\land_A]$.
\end{prop}

Note that the requirement $\top \in \D(A)$ is needed for this proposition to
hold: we have $\top \in \MLoLa^{+}(A)$, so for the proposition to be true there
has to be a formula $\delta \in \D(\funP(A))$ such that $\top \oneseq
\delta[\land_A]$.
Since $\top \in \D(\funP(A))$ we can take $\delta = \top$ --- without having
$\top \in \D(\funP(A))$ there would be no guarantee in general that an
appropriate formula $\delta$ can be found.

We have seen that nabla formulas are disjunctive,  and we shall soon see that
they provide disjunctive bases for many modal signatures.
In order for the approach of constructing disjunctive bases via nabla formulas
to work however, it is necessary that the ambient functor $\funT$ preserves weak
pullbacks.
This is one sense in which disjunctive bases generalize nabla formulas: we shall
soon see an example of a functor which does not preserve weak pullbacks, but
does have a signature admitting a disjunctive basis.
However, we want to stress that the main motivation behind disjunctive bases is
not to ``get by without weak pullback preservation'' --- in fact, most natural
examples of disjunctive bases we can think of apply to functors that do
preserve weak pullbacks.
Rather, the point is that the existence of a disjunctive basis is a property of
\emph{modal signatures}, not of  functors.
The choice of a modal signature for a functor is not arbitrary, it can
substantially affect the properties of the modal languages that we get.
In particular, even for functors that do preserve weak pullbacks, we can have
natural choices of modal signatures for which the nabla formulas associated
with the functor do not provide a disjunctive basis.
Typically, this will happen when the modal signature is (a) infinite, and (b)
not expressively complete.
An example of this situation that will be given special attention is graded
modal logic, to be presented in the next subsection.
But there are much simpler cases, as we shall see.

\subsection{Examples}

In this subsection we present a number of examples of modal signatures admitting
disjunctive bases, and one example of a modal signature that provably does not
admit a disjunctive basis.

\subsubsection{Disjunctive bases for weak pullback preserving functors}
In the previous section we noted that disjunctive formulas generalize the Moss
modalities.
In many interesting cases this suffices to find a disjunctive basis.

\begin{prop}%
\label{p:wpp-db}
Let $\La$ be a signature for a weak-pullback preserving functor $\funT$.
If $\La$ is Lyndon complete, then it admits a disjunctive basis.
\end{prop}

\begin{proof}
We shall use the (infinitary) nabla based logic for $\funT$ as an auxiliary
language.
Let $\MLone^\infty_{\nb}(A)$ be defined by the  grammar:
\[
    \al ::= \top \mid \bigvee \Phi \mid \nb \Ga
\]
where $\Phi \subseteq \MLone^\infty_{\nb}(A) $ and $\Ga \in \funT A$.
Semantics of formulas in $\MLone^\infty_{\nb}(A)$ in a one-step model is defined
by the obvious recursion, where the interpretation of $\nb \Ga$ is as before.

We first claim that for every formula $\al \in \MLone^\infty_{\nb}(A) $, where
$A = \{a_1,\dots,a_n\}$, there is a monotone $n$-place predicate lifting $\la$
for $\funT$ such that  $\al$ is equivalent to the formula $\hs_\la (a_1,\dots,a_n)$.
This follows from results in~\cite{leal:pred08}, together with the easy
observation that (monotone) predicate liftings are closed under arbitrary
disjunctions.
We can safely assume that $\hs_\la (a_1,\dots,a_n) \in \MLoneLa^+(A)$ since $\La$
was assumed to be Lyndon complete.%
\footnote{%
    Strictly speaking Lyndon completeness only guarantees there exists some formula
    in $\MLoneLa^+(A)$ that is equivalent to $\hs_\la (a_1,\dots,a_n)$, but taking
    care with this distinction would only complicate notation.
}

With this in mind, let $\D_{\nb}(A)$ be the set of all finite disjunctions of
formulas of the form $\top$ or of the form $\hs_\la (a_1,\dots,a_n)$, where $\la$
is the $n$-place predicate lifting associated with some $\al \in
\MLone^\infty_\nb(A)$.
As mentioned, all formulas of the form $\nb\Ga$ are
disjunctive, and since disjunctivity is closed under taking arbitrary
disjunctions, all formulas in $\MLone^\infty_\nb(A)$ are disjunctive too ---
hence all formulas in  $\D_{\nb}(A)$ are disjunctive.

It remains to prove that $\D_{\nb}(A)$ is a basis for $\La$, so we need to show
that $\D_{\nb}(A)$ is distributive over $\La$ and admits a binary distributive
law. For the purpose of proving distributivity over $\La$ it suffices to show
that any formula $\beta \in \MLoLa^{+}(A)$ is equivalent to a formula
$\al[\wedge_{A}]$, where $\al \in \MLone^\infty_\nb(\funP A)$.
In other words we want to prove, for an arbitrary formula $\beta \in
\MLoLa^{+}(A)$:
\begin{equation}
\label{eq:nbdb}
\beta \equiv^{1} \bv \{ (\nb \Ga)[\wedge_A] \mid \Ga \in \funT \funP A \; \& \; \funP A, \Ga, \id \satone \beta
\},
\end{equation}
where $\id : B \mapsto B$ denotes the canonical marking on the set $\funP A$.
The notation here may need a bit of explanation: given $\Ga \in \funT \funP A$,
we can apply the functor $\funT$ to the map $\wedge_A : \funP A \to \Latt (A)$
and apply this to $\Gamma$ to obtain $(\funT \wedge_A)\Gamma \in \funT \Latt
(A)$.
With this established we define:
\[
    (\nb \Ga )[\wedge_A] := \nb (\funT \wedge_A)\Ga.
\]

For a proof of the left-to-right direction of~\eqref{eq:nbdb}, assume that
$S,\si,m \satone \beta$.
From this it follows by invariance (Proposition~\ref{p:1invar}) that $\funP A,
(\funT m)\si, \id \satone \beta$, so that $\nb \Ga$ where $\Ga \isdef (\funT m) \si \in \funT\funP
A$ provides a candidate disjunct on the right hand side of~\eqref{eq:nbdb}.
It remains to show that $S,\si,m \satone \nb (\funT m) \si[\wedge_A]$, but this
is immediate by definition of the semantics of $\nb$.

For the opposite direction of~\eqref{eq:nbdb}, let $\Ga \in \funT\funP A$ be
such that $\funP A, \Ga, \id \satone \beta$.
In order to show that $\beta$ is a one-step semantic consequence of
$\nb \Ga[\wedge_A]$, let $(S,\si,m)$ be a one-step model such that
$S,\si,m \satone \nb(\funT\wedge_{A})\Ga$.
Without loss of generality we may assume that $(S,\si,m) = (\funP A,\De,\id)$
for some $\De \in \funT\funP A$.

By the semantics of $\nb$ it then follows from $\funP A,\De,\id  \satone
\nb\Ga[\wedge_A]$ that $(\De, (\funT\wedge_{A})\Ga) \in \ol{\funT}(\satzero)$.
But since $(B,\wedge_{A}(C)) \in {\satzero}$ implies that $C \sse B$,
% by Fact~\ref{f:rl}
we easily obtain that $(\Ga,\De) \in \ol{\funT}(\sse)$. We can now apply the following claim, the proof of which we leave as an exercise:

\begin{claimfirst}
Let $(S,\si,m)$ and ($S',\si',m')$ be two one-step models, and let $Z \sse S
\times S'$ be a relation such that $(\si,\si') \in \ol{\funT}Z$, and $m(s) \sse
m'(s')$, for all $(s,s') \in Z$.
Then for all $\al \in \MLoLa^{+}(A)$:
\[
S,\si,m \satone \al \text{ implies } S',\si',m' \satone \al.
\]
\end{claimfirst}
It is easy to see that the claim is applicable to the one-step models
$(\funP A,\Ga,\id)$ and $(\funP A,\De,\id)$, and the relation $Z = \;\sse$: we already showed that $(\Ga,\De) \in \ol{\funT}(\sse)$. Furthermore, if  $B \sse B'$ then trivially $\id(B) \sse \id(B')$.
Hence it follows from $\funP A, \Ga, \id \satone \beta$ that $\funP A, \De, \id
\satone \beta$.

For the binary distributive law, we leave it to the reader to check that for $\Ga \in \funT A$ and $\Ga' \in \funT B$, the conjunction $\nb \Ga \wedge \nb \Ga'$ is equivalent to the possibly infinite disjunction of all formulas of the form $\nb \Ga'' [\theta_{A,B}]$ such that $\Ga''\in \funT (A \times B)$, $\funT \pi_A(\Ga'') = \Ga$ and $\funT \pi_B(\Ga'') = \Ga'$. With this claim in place, the binary distributive law is easily established.
\end{proof}

\subsubsection{A simple example}

The simplest non-trivial example of a disjunctive basis we can think of, that
is not a special case of the nabla-based approach of the previous paragraph,
is the following.
Consider the functor $\funT = \Sigma \times \funI$, where $\Sigma$ is a
countably infinite alphabet. For this functor, $\funT$-coalgebras are (up to
unfolding) just $\Sigma$-streams, or infinite words for which the  alphabet is
contained in $\Sigma$.
They can be viewed as triples $(S,\sigma_1,\sigma_2)$ where $\sigma_1 : S \to
\Sigma$ and $\sigma_2 : S \to S$.
A simple and reasonably natural modal signature $\La$ for this functor is the
following: we have one nullary modality $!l$ for each $l \in \Sigma$ with the
interpretation: $\bbS,s \sat {!l}$ iff $\sigma_1(s) = l$, and we have a single
one-place modality $\nxt{}$ with the interpretation $\bbS,s \sat \nxt{}\phi$
iff $\bbS,\sigma_2(s) \sat \phi$.
We can easily describe this language as a modal signature $\Lambda$ for the
functor $\funT$, and it is not hard to show that $\Lambda$ has a disjunctive
basis consisting of the formulas $!l$ for $l \in \Sigma$ and $\nxt{}a$ for
$a \in A$.

The one-step nabla formulas for $\funT$, which are of the form $\nb (l,a)$, with
$(S, (l',s),m) \satone \nb(l,a)$ iff $l' = l$ and $s \satzero a$, can easily
be expressed as $\La$-formulas by writing $\nb (l,a) \isdef {!l} \land \nxt{} a$.
However, these formulas do not form a disjunctive basis for $\La$. To see why, just consider the set of variables $ A = \{a\}$. The property of distributivity of disjunctive bases would then require that the formula $\nxt{}a$ should be equivalent to a finite disjunction of formulas of the form:
\[ {!l} \land \nxt{} a \]
 where $l \in \Sigma$. But the formula $\nxt{}a$ corresponds to an infinite disjunction:
\[
\bigvee \{{!l} \land \nxt{} a \mid l \in \Sigma\}
\]
It is fairly obvious that this is not equivalent to any finitary disjunction of the required shape.
% Note that each nabla formula $\nabla(l,a)$ is equivalent to the formula
% $!l \wedge \nxt{}a$.
It is not hard to come up with other, similar examples.

\subsubsection{Graded modal logic}

A much more involved example of a modal logic that admits a disjunctive basis
which cannot be reduced to nabla formulas,  is graded modal logic.
The bag functor (defined in Section~\ref{s:someexamples}) does preserve weak pullbacks, and so its Moss modalities are
disjunctive, and so the set of all monotone liftings for $\funB$ does admit a
disjunctive basis as an instance of Proposition~\ref{p:wpp-db}.
Note, however, that the latter proposition does not apply to graded modal logic,
since its signature $\Si_{\funB}$ is not expressively complete;
this was essentially shown in~\cite{pacu:majo04}.
% by Pacuit \& Salame
It was observed  in~\cite{berg:moss09} that very simple formulas in the
one-step language $\MLone_{\Si_{\funB}}$ are impossible to express in the
(finitary) Moss language; consequently, the Moss modalities for the bag functor
are not suitable to provide disjunctive normal forms for graded modal logic.
Still, the signature $\Si_\funB$ does have a disjunctive basis.

\begin{definition}
We say that a one-step model for the finite bag functor is \emph{Kripkean}
if all states have multiplicity $1$.
Note that a Kripkean one-step model $(S,\si,m)$ can also be seen as a structure
(in the sense of standard first-order model theory) for a first-order signature
consisting of a monadic predicate for each $a \in A$:
Simply consider the pair $(S,V_{m})$, where $V_{m}: A \to \funP S$ is the
interpretation given by putting $V_{m}(a) \isdef \{ s \in S \mid a \in m(s) \}$.
We consider special basic formulas of monadic first-order logic of the form:
\[
\be(\ol{a},B) \isdef
\exists \ol{x} (\diff(\ol{x}) \wedge \bw_{i \in I} a_i(x_i)
   \wedge \forall y (\diff(\ol{x},y) \rightarrow \bv_{b \in B}b(y))),
\]
where $\diff(x_{1},\ldots,x_{n})$ abbreviates the formula
$\diff(x_{1},\ldots,x_{n}) = \bigwedge_{1\leq i<j\leq n} x_{i} \neq x_{j}$. To be explicit, we allow the case where $\ol{a}$ is the empty sequence $\epsilon$, in which case we get:
\[
    \be(\epsilon,B) = \forall x  \bv_{b \in B}b(x).
\]
It is not hard to see that any Kripkean one-step $\funB$-model $(S,\si,m)$
satisfies:
\begin{equation}
\label{eq:nbB1}
S,\si,m \satone \be(\ol{a},B) \text{ implies }
S,\si,m' \satone \be(\ol{a},B) \text{ for some } m' \sse m
   \text{ with } \Ran(m') \sse \funP_{\leq 1}A.
\end{equation}

We can turn the formula $\be(\ol{a},B)$ into a modality $\disbag{\ol{a}}{B}$
that can be interpreted in \emph{all} one-step $\funB$-models, using the
observation that every one-step $\funB$-frame $(S,\si)$ (with $\si: S \to\om$)
has a unique Kripkean cover $(\ti{S},\ti{\si})$ defined by putting
$\ti{S} \isdef \bigcup\{ s \times \si(s) \mid s \in S\}$,
and $\ti{\si} \in \funB \ti{S}$ is defined by $\ti{\si}(s,i) \isdef 1$ for all $s \in S$ and $i \in \si(s)$ (where we view
each finite ordinal as the set of all smaller ordinals).
% , so in particular $0$ is defined to be the empty set).
Then we can define, for an arbitrary one-step $\funB$-model $(S,\si)$
\begin{equation}
\label{eq:Bnb}
S,\si,m \satone \disbag{\ol{a}}{B} \text{ if }
   \ti{S},\ti{\si},m\circ\pi_{S} \satone \be(\ol{a},B),
\end{equation}
where $\pi_{S}$ is the projection map $\pi_{S}: \ti{S} \to S$.
It is then an immediate consequence of~\eqref{eq:nbB1} that $\disbag{\ol{a}}{B}$
is a disjunctive formula.
\end{definition}
Our main aim in this section is to show that the modalities $\disbag{\ol{a}}{B}$
provide a disjunctive basis for the signature $\SigB$.
As far as we know, this result is new.
The hardest part in proving it is to show that the modalities
$\disbag{\ol{a}}{B}$ can be expressed as one-step formulas in
$\MLone_{\SigB}^+(A)$.
The reason that this is not so easy is subtle; by contrast, it is fairly
straightforward to show that formulas of the form $\disbag{\ol{a}}{B}$ can be
expressed in $\MLone_{\SigB}(A)$, using Ehrenfeucht-Fra\"{\i}ss{\'e} games, see
e.g.~Fontaine \& Place~\cite{font:fram10}.
However, a proper disjunctive basis as we have defined it has to consist of
\emph{positive} formulas, and this will be crucial for applications to modal
fixpoint logics%
\footnote{The same subtlety appears in Janin \& Lenzi~\cite{jani:rela01},
   where the translation of the language $\D_{\funB}$ into
   $\MLone_{\SigB}^+$ is required to prove that the graded $\mu$-calculus is
   equivalent, over trees, to monadic second-order logic.
   Proposition~\ref{p:jan} in fact fills a minor gap in this proof.
   }.

\begin{prop}%
\label{p:jan}
Every formula $\disbag{\ol{a}}{B}$ is one-step equivalent to a
formula in $\MLone^{+}_{\SigB}(A)$.
\end{prop}

Our main tool in proving this proposition will be Hall's Marriage Theorem,
which can be formulated as follows.
A \emph{matching} of a bi-partite graph $\bbG = (V_1,V_2,E)$ is a subset $M$ of
$E$ such that no two edges in $M$ share any common vertex.
$M$ is said to \emph{cover} $V_1$ if $\Dom M = V_1$.

\begin{myfact}[Hall's Marriage Theorem]%
\label{f:HMT}
Let $\bbG$ be a finite bi-partite graph, $\bbG = (V_1,V_2,E)$.
Then $\bbG$ has a matching that covers $V_1$ iff, for all $U \sse V_1$,
$\sz{U} \leq \sz{E[U]}$, where $E[U]$ is the set of vertices
in $V_2$ that are adjacent to some element of $U$.
\end{myfact}

\noindent
\begin{proofof}{Proposition~\ref{p:jan}}
We will show this for the simple case where $B$ is a singleton $\{b\}$.
The general case is an immediate consequence of this (consider the
substitution $B \mapsto \bv B$).

Let  $\ol{a} = (a_{1},\ldots,a_{n})$. Define $I \isdef \{ 1, \ldots, n \}$.
For each subset $J \sse I$, let $\chi_{J}$ be the formula
\[
\chi_{J} \isdef
\dia^{\sz{J}} \bigvee_{i \in J} a_i
\wedge \Box^{n +1 - \sz{J}} (\bigvee_{i \in J} a_i \vee b),
\]
and let $\gamma$ be the conjunction
\[
\gamma \isdef \bw \{\chi_J \mid J \sse I \}.
\]
% where as a convention we set $\dia^{k - 1}$ to be $\top$ for $k = 0$.
What the formula $\chi_J$ says about a Kripkean (finite) one-step model is
that at least $\sz{J}$ elements satisfy the disjunction of the set
$\{a_i \mid i \in J\}$, while all but at most $n - \sz{J}$ elements satisfy
the disjunction of the set $\{a_i \mid i \in J\}\cup \{b\}$.
Abbreviating $\disbag{\ol{a}}{b} \isdef \disbag{\ol{a}}{\{b\}}$, we claim that
\begin{equation}
\label{eq:bagnb}
\ga \equiv^{1} \disbag{\ol{a}}{b},
\end{equation}
and to prove this it suffices to consider Kripkean one-step models.

It is straightforward to verify that the formula $\gamma$ is a semantic one-step
consequence of $\disbag{\ol{a}}{b}$.
For the converse, consider a Kripkean one-step model $(S,\si,m)$ in which
$\gamma$ is true.
Let $K$ be an index set of size $\sz{S} - n$, and disjoint from the index set
$I = \{1,\dots,n\}$.
Clearly then, $\sz{I \cup K} = \sz{I} + \sz{K} = \sz{S}$.
Furthermore, let $a_{k} \isdef b$, for all $k \in K$.
To apply Hall's theorem, we define a bipartite graph
$\bbG \isdef (V_{1},V_{2},E)$ by setting
$V_{1}  \isdef  I \cup K$,
$V_{2}  \isdef  S$, and
$E      \isdef  \{ (j,s) \in (I \cup K) \times S \mid a_{j} \in m(s) \}$. Note that this graph is finite: by definition of the bag functor only finitely many elements may have non-zero multiplicity in a one-step model, so every Kripkean one-step model has to be finite. Hence $S$ is finite, and so both $I$ and $K$ are finite since  $\sz{I \cup K} = \sz{S} $.

\begin{claimfirst}%
\label{cl:gml1}
The graph $G$ has a matching that covers $V_{1}$.
\end{claimfirst}

\begin{pfclaim}
We check the Hall marriage condition for an arbitrary subset $H \sse V_{0}$.
In order to prove that the size of $E[H]$ is greater than that of $H$ itself,
we consider the formula $\chi_{H \cap I}$.
We make a case distinction.
\medskip

\noindent
\textbf{Case 1:} $H \sse I$.
Then $\chi_{H \cap I} = \chi_{H}$ implies $\dia^{\sz{H}} \bigvee_{i \in H} a_i$.
This means that at least $\sz{H}$ elements of $S$ satisfy at least one variable
in the set $\{a_i \mid i \in H\}$.
By the definition of the graph $\bbG$, this is just another way of saying that
$\sz{H} \leq \sz{E[H]}$, as required.
\medskip

\noindent
\textbf{Case 2:} $H \cap K \neq \nada$.
Let $J \isdef H \cap I$, then the formula $\chi_{H \cap I} = \chi_{J}$ implies
the formula
\[
\Box^{n+1 - \sz{J}} (\bigvee_{j \in J} a_j \vee b).
\]
Now, if $s \in S$ satisfies either $b$ or some $a_j$ for $j \in J$, then by the
construction of $\bbG$ we have $s \in E[H]$.
We now see that $\sz{S \setminus E[H]} \leq n - \sz{J}$. Hence we get:
\[
\sz{E[H]} \geq \sz{S} - (n - \sz{J}) = \sz{S} - n +\sz{J}.
\]
% and so we need to show that $\sz{H} \leq \sz{S} - n + \sz{J}$.
But note that $H = J \cup (H \cap K)$, so that we find
\[
\sz{H} \leq \sz{J} + \sz{H \cap K} \leq \sz{J} + \sz{K}
= \sz{J} + (\sz{S} - n),
\]
From these two inequalities it is immediate that $\sz{H} \leq \sz{E[H]}$, as
required.
\end{pfclaim}

Now consider a matching $M$ that covers $V_1$.
Since the size of the set $V_1$ is the same as that of $V_2$, any matching $M$
of $\bbG$ that covers $V_1$ is (the graph of) a bijection between these two sets.
Furthermore, it easily follows that such an
% if $M$ has an edge from  $i \in I$ to $s \in S$, then $a_{i} \in
% m(s)$, and if $M$ has an edge from $j \in K$ to $s \in S$, then by definition of
% $\bbG$ we have $b \in m(s)$.
% It follows that
$M$ restricts to a bijection between $I$ and a subset
$\{s_1,\dots,s_n\}$ of $S$ such that $a_i \in m(s_{i})$ for each $i \in I$, and
that $b \in m(t)$ for each $t \notin \{u_1,\dots,u_n\}$.
Hence $\disbag{\ol{a}}{b}$ is true in $(S,\si,m)$, as required.
\end{proofof}
\medskip

In light of this proposition, we shall continue to use the notation
$\disbag{\ol{a}}{B}$ for the equivalent formula in $\MLone^{+}_{\SigB}(A)$
provided in the proof.

\begin{definition}
We define $\D_{\funB}(A)$ by the following grammar:
\[ \delta ::= \top \mid \disbag{\ol{a}}{B} \mid \delta \vee \delta \]
where $\ol{a}$ is a tuple of elements from $A$ and $B \subseteq A$.
\end{definition}

\begin{thm}%
\label{t:bagdisbas}
The assignment $\D_{\funB}$ provides a disjunctive basis for the signature
$\SigB$.
\end{thm}

\begin{proof}
It remains to prove that  $\D_{\funB}$ is distributive over $\SigB$, and admits
a binary distributive law.
For the first part, consider the formula $\dia^k \pi  \in \MLone^{+}_{\SigB}(A)$.
Note that $\nada$ is a variable in $\funP A$ and that $\bigvee \{\nada\}
[\wedge_A] = \bigwedge \nada = \top$.
With this in mind, it is not hard to see that $\dia^k\pi$ can be rewritten
equivalently as:
\[
\dia^k \pi = \bigvee \{\disbag{B_1,\dots,B_k}{\{\nada\}}
\mid \bigwedge B_i \vDash \pi, \text{ for each } i \}[\wedge_A]
\]
Here, $\vDash$ denotes propositional consequence between formulas in $\Bool(A)$.
Next, consider the formula $\Box^k \pi \in \MLone^{+}_{\SigB}(A)$. Keeping in mind that $\bigwedge \nada = \top$, the reader can verify that this is equivalent to:
\[
    \bigvee_{m < k} \disbag{\underbrace{\nada,\dots,\nada}_{m \text{ times}}}{ \{B \subseteq A \mid \bigwedge B \vDash \pi\}}[\wedge_A]
\]
To establish the binary distributive law, let $\delta \in \D_{\funB}(A)$ and $\delta' \in \D_\funB(B)$. Then $\delta$ is of the form $\alpha_1 \vee \dots \vee \alpha_n$ and $\delta'$ is of the form $\beta_1 \vee \dots \vee \beta_m$, where each $\alpha_i$ is either equal to $\top$ or of the form $\disbag{\ol{a}}{A'}$ for some $\ol{a},A'$ and each  $\beta_j$ is either equal to $\top$ or of the form $\disbag{\ol{b}}{B'}$ for some $\ol{b},B'$. By distributing the conjunction over disjunctions, we can rewrite the formula $\delta \wedge \delta'$ as an equivalent disjunction of formulas $\alpha_i \wedge \beta_j$ for $1 \leq i \leq n$ and $1 \leq j \leq m$. So it suffices to show that the required distributive law holds for conjunctions of this shape.

In the case where $\alpha_i = \beta_j = \top$, we have  $\alpha_i \wedge \beta_j \oneseq \top$. But $\top \in \D(A \utimes B)$ and $\top = \top[\theta_{A,B}]$, so this case is done. If  one and only one of $\alpha_i,\beta_j$ is equal to $\top$, say $\alpha_i = \disbag{\ol{a}}{A'}$ and $\beta_j = \top$, then:
\[
    \alpha_i \wedge \beta_j \oneseq \disbag{\ol{a}}{A'} \in \D(A) \subseteq \D(A \utimes B).
\]
But  $\disbag{\ol{a}}{A'}[\theta_{A,B}] = \disbag{\ol{a}}{A'}$, so we are done in this case as well.

Finally, the interesting case is where $\alpha_i$ is of the form $\disbag{a_1,\dots,a_k}{A'}$ and $\beta_j$ is of the form $\disbag{b_1,\dots,b_l}{B'}$. For this case, we need some definitions: we fix the formulas $\alpha_i = \disbag{a_1,\dots,a_k}{A'}$ and $\beta_j = \disbag{b_1,\dots,b_l}{B'}$. Let an \emph{overlap record} be a subset $O \subseteq \{1,\dots,k\} \times \{1,\dots,l\}$ such that $O$ is the graph of a bijection from some subset of $\{1,\dots,k\}$ onto some subset of $\{1,\dots,l\}$. We denote the set of such overlap records by $\overlap$. Given an overlap record $O$, an \emph{$O$-pair} is a pair $(c_1,c_2)$ of functions $c_1 : \{1,\dots,k\}\setminus \pi_1[O] \to B'$ and $c_2 : \{1,\dots,l\}\setminus \pi_2[O] \to A'$. Finally, a \emph{case description} is a triple $(O,c_1,c_2)$ where $O$ is an overlap record and $(c_1,c_2)$ is an $O$-pair. Note that there are finitely many case descriptions.

Given a case description $(O,c_1,c_2)$, let $\{1,\dots,k\}\setminus \pi_1[O] = \{p_1,\dots,p_s\}$, let $ \{1,\dots,l\}\setminus \pi_2[O] = \{q_1,\dots,q_t\}$, and let $\vec{O}$ be a list of the pairs in $O$ in some arbitrary order. Now define the formula  $\chi(O,c_1,c_2)$ to be:
\[
    \disbag{(a_{p_1},c_1(p_1)),\dots,(a_{p_s},c_1(p_s)),(b_{q_1},c_2(q_1)),\dots,(b_{q_t},c_2(q_t)),\vec{O}\;}{\;A'\times B'}
\]
This is a formula in $\D_{\funB}(A\utimes B)$, in fact it is a formula in  $\D_{\funB}(A\times B)$.
It is a fairly straightforward task to verify that the formula $\alpha_i \wedge
\beta_j$ is equivalent to the disjunction:
\[
    \bigvee \{\chi(O,c_1,c_2) \mid (O,c_1,c_2) \text{ is a case description}\}
\]
We leave this as an exercise for the reader.
\end{proof}

\subsubsection{An example without weak pullback preservation}
Here is an example of a functor that does not preserve weak pullbacks, but still has a natural modal signature that admits a disjunctive basis.
Let $\mathsf{F}$ be the subfunctor of $\funP^2$  given by setting $\mathsf{F} X$ to be the set of pairs $(Y,Z) \in {(\funP X)}^2$ such that at least one of the sets $Y,Z$ is finite. That is:
\[
    \mathsf{F}X = (\funP_\omega X \times \funP X) \cup (\funP X \times \funP_\omega X).
\]
 This is a well defined subfunctor of $\funP^2$ since the direct image of a finite subset $Y \subseteq X$ under any given map $f : X \to X'$ is finite.

$\mathsf{F}$ does not preserve weak pullbacks. To see why, consider the constant map $c : \om \to \{0\}$, and consider the following two objects in $\mathsf{F} \om$: $\alpha = (\{0\},\om)$ and $\beta = (\om,\{0\})$. Clearly, $\mathsf{F}c(\alpha) = \mathsf{F}c(\beta) = (\{0\},\{0\})$. But given that $(P,p_1,p_2)$ is the pullback of the diagram: $\om \longrightarrow \{0\} \longleftarrow \om$ where both maps are equal to  $c$, if $\mathsf{F}$ preserves weak pullbacks then  there should be some  pair $\chi = (Y,Z) \in \mathsf{F}P$  such that $\mathsf{F}p_1(\chi) = \alpha$  and $\mathsf{F}p_2(\chi) = \beta$. But then $p_1[Z] = \om$ and likewise $p_2[Y] = \om$.
So  both $Y$ and $Z$ would have to be infinite, contradicting our definition of $\mathsf{F}$.

Consider the modal signature consisting of the usual labelled diamond modalities $\Diamond^0$ and $\Diamond^1$, quantifying over the left and right set in a pair $(Y,Z) \in \mathsf{F}X$ respectively, and their dual box modalities. Then we can define corresponding (finitary) nablas $\nabla^0$ and $\nabla^1$ in the usual way. Then, finite disjunctions of $\top$ and conjunctions of the form $\nabla^0 A_0 \wedge \nabla^1 A_1$ for finite sets $A_0,A_1$ will provide a disjunctive basis for this signature. The reason that the formulas $\nabla^0 A_0 \wedge \nabla^1 A_1$ are still disjunctive is this: given a one-step model $(X,(Y,Z),m)$ satisfying the formula $\nabla^0 A_0 \wedge \nabla^1 A_1$, where $(Y,Z) \in \mathsf{F}X$ and $m$ is a marking, we can construct an appropriate one-step frame morphism $f : (X',(Y',Z')) \to (X,(Y,Z))$ and a marking $m'$ on $X'$ as follows. Take:
\begin{eqnarray*}
X' & := &  (X \times A_0) \cup (X \times A_1), \\
  Y' & := &  \{(u,a) \in Y \times A_0 \mid a \in m(u) \}, \\
 Z' & := & \{(u,a) \in Z \times A_1 \mid a \in m(u) \}.
\end{eqnarray*}
The map $f$ is of course just the projection to $X$, and the marking $m'$ is defined by $m'(u,a) = \{a\}$. The crucial point is that since one of $Y,Z$ is finite, and since $A_0,A_1$ are both finite sets, one of the sets $Y',Z'$ is finite too, so $(Y',Z') \in \mathsf{F}X'$ as required. The proof that the formula $\nabla^0 A_0 \wedge \nabla^1 A_1$ also holds in the one-step model $(X',(Y',Z'),m')$ is routine.

\subsubsection{Monotone modal logic: a signature without disjunctive basis}

Finally, we provide an example of a signature that does not admit any
disjunctive basis:
% For lack of space, the proof of this result is deferred to the appendix.

\begin{prop}%
\label{p:mon-nodb}
The signature $\SigM$ consisting of the box- and diamond liftings for $\funM$
does not have a disjunctive basis.
\end{prop}

\begin{proof}
Let $L$ be the standard relation lifting for the monotone neighborhood functor.
Given two one-step models $X,\xi,m$ and $X',\xi',m'$ over a set of variables $A$,
we write  $ u \preceq u'$ if $m(u) \sse m'(u')$ for $u \in X$ and $u' \in
X'$, and we say that $X',\xi',m'$ \emph{simulates} $X,\xi,m$ if $(\xi,\xi') \in
L(\preceq)$.
A straightforward proof will verify the following claim.

\begin{claimfirst}%
\label{cl:presprop}
If $X',\xi',m'$ simulates $X,\xi,m$ then for every one-step formula $\alpha \in
\MLoneLa^{+}(A)$, $X,\xi,m\satone \alpha$ implies $X',\xi',m' \satone \alpha$.
\end{claimfirst}

Given a set $A$, let $\eta_A : A \to \funP{A}$ denote the map given by the unit of
the powerset monad, i.e., it is the singleton map $\eta_A : a \mapsto \{a\}$.
Furthermore, recall that $\land_A$ is the subsitution mapping $B \in \funP A$
to $\bw B$.

\begin{claim}%
\label{cl:theprop}
Let $\alpha$ be any one-step formula in $\MLoneLa(\funP A)$ and let $(X,\xi,m)$
be a one-step model with $m : X \to \funP{A}$.
Consider the map $\eta_{\funP A} : \funP A \to \funP \funP A$, so that $\eta_{\funP A}
\circ m$ is a marking of $X$ with variables from $\funP A$.
\begin{enumerate}
\item If $X,\xi, \eta_{\funP A} \circ m \satone \alpha $ then $X ,\xi ,m \satone
\alpha[\land_A]$.
\item If $X,\xi,m \satone \alpha[\land_A]$ and the empty set does not appear as
a variable in $\alpha$, and furthermore $m(u)$ is a singleton for each $u \in X$,
then $X,\xi, \eta_{\funP A} \circ m \satone \alpha $.
\end{enumerate}
\end{claim}

\begin{pfclaim}
For the first part of the proposition, it suffices to note that
$\mngone{B}_{\eta_{\funP A} \circ m} \sse \mngone{\bw B}_m$ for each
$B \in \funP A$, and the result then follows by monotonicity of the predicate
lifting corresponding to the one-step formula $\alpha$.

For the second part, it suffices to note that under the additional constraint
that $m(u)$ is a singleton for each $u \in X$ and the empty set does not appear
as a variable in $\alpha$,  we have $\mngone{\bw B}_m \sse
\mngone{B}_{\eta_{\funP A} \circ m}$ for each $B \in \funP A$ that appears as a
variable in $\alpha$.
To prove this,  suppose that $u \in \mngone{\bw B}_m$.
Since $B$ appears in $\alpha$ it is non-empty, and since $m(u)$ is a singleton,
say $m(u) = \{b\}$, it follows that we must in fact have $B = \{b\}$. Hence:
\[
    B \in \{\{b\}\} = \{m(u)\} = \eta_{\funP A}(m(u))
\]
so $u \in \mngone{B}_{\eta_{\funP A}\circ m}$ as required.
\end{pfclaim}

Now, let $A = \{a,b,c\}$ and consider the formula $\psi = \nabla \{\{a,b\},\{c\}
\}$.
If $\MLoneLa$ admits a disjunctive basis, then there is a disjunctive formula
$\delta$  in $\MLoneLa(\funP A)$ such that $\psi = \delta[\land_A]$.

So suppose $\delta \in \MLoneLa(\funP A)$ is disjunctive, and suppose that
$\psi = \delta[\land_A]$.
We may in fact assume w.l.o.g.\ that the empty set does not appear as a variable
in $\delta$, since otherwise we just use instead the formula $\delta[\top/\nada]$,
which is still disjunctive (this is easy to prove).
We have $\delta[\top/\nada][\land_A] = \delta[\land_A]$ since $\land_A(\nada) =
\bw \nada = \top$.

With this in mind, consider the one-step model $X,\xi,m$ where $X = \{x_1,x_2,
x_3\}$, $\xi = \{\{x_1,x_2\},\{x_3\},X\}$ and $m(x_1) = \{a\}$, $m(x_2) = \{b\}$
and $m(x_3) = \{c\}$.
It is easy to see that $X,\xi,m \satone \psi$, so by assumption $X,\xi,m\satone
\delta[\land_A]$.
But since the marking $m$ maps every element of $X$ to a singleton, item 2 of
Claim~\ref{cl:theprop} gives us that  $X,\xi,\eta_{\funP A} \circ m \satone
\delta$.

Now, define a new one-step model $X,\xi,h$ where as before $X = \{x_1,x_2,x_3\}$
and $\xi = \{\{x_1,x_2\},\{x_3\},X\}$ but where the marking $h : X \to \funP
\funP A$ (with respect to variables in $\funP A$) is defined by setting $h(x_1)
= \{\{a\}\}$, $h(x_2) = \{\{a\},\{c\}\}$ and $h(x_3) = \{\{c\}\}$.
It is a matter of simple verification to check that $X,\xi,h$ in fact simulates
$X,\xi,\eta_{\funP A} \circ m$, so by Proposition~\ref{cl:presprop} we get
$X,\xi,h \satone \delta$.

Since $\delta$ is disjunctive, there should be a one-step model $X',\xi',h'$ and
a map $f : X' \to X$ such that:
$X',\xi',h' \satone \delta$,
$\funM f(\xi') = \xi$,
$h'(u) \sse h(f(u))$ for all $u \in X'$ and
$h'(u)$ is at most a singleton for each $u \in X'$.
By monotonicity of $\delta$ we can in fact assume w.l.o.g.\ that $h'(u)$ is
precisely a singleton for each $u \in X'$: if $h'(u) = \nada$, just pick some
element $e$ of $h(f(u))$ (since $h(v)$ is non-empty for each $v \in X$) and set
$h'(u) = \{e\}$.
The resulting marking still satisfies all the conditions above.

But this means that we can define a marking $n : X' \to \funP{A}$ by taking each
$n(u)$ for $u \in X'$ to be the unique $B \sse A$ such that $h'(u) = \{B\}$.
Clearly, $h' = \eta_{\funP A} \circ n$, so by the first part of
Claim~\ref{cl:theprop}, we get $X',\xi',n \satone \delta[\land_A]$, hence
$X',\xi',n \satone \psi$, i.e. $X',\xi',n \satone \nabla \{\{a,b\},\{c\}\}$.
But from the definition of the marking $h$, the condition that $h'(u) \sse
h(f(u))$ for all $u \in X'$ and from the definition of $n$ it is clear that, for
all $u \in X'$, we have $n(u) = \{a\}$ or $n(u) = \{c\}$.
So to finally reach our desired contradiction, it suffices to prove the
following.

\begin{claim}
Let $X,\xi,m$ be any one-step model such that $X,\xi,m \satone \nabla \{\{a,b\},
\{c\}\}$.
Then either there is some $u \in X$ with $\{a,c\} \sse m(u)$, or there is some
$u \in X$ with $b \in m(u)$.
\end{claim}

\begin{pfclaim}
Suppose there is no $u \in X$ with $b \in m(u)$.
Then there is some set $Z \in \xi$ such that every $v \in Z$ satisfies $a$.
Furthermore there must be some $B \in \xi$ such that every $l \in B$ is
satisfied by some member of $Z$.
The only choice possible for this is $\{c\}$, hence some member of $Z$ must
satisfy both $a$ and $c$.
\end{pfclaim}

This finishes the proof of Proposition~\ref{p:mon-nodb}.
\end{proof}

Interestingly, we may ``repair'' the failure of admitting a disjunctive basis
by moving to the \emph{supported companion functor} $\underline{\funM}$.
This functor is defined as the subfunctor of $\funP \times \funM$, given (on
objects) by $\underline{\funM}S \isdef \{ (U,\si) \in \funP S \times \funM S
\mid U \text{ supports } \si \}$, where the
notion of support is defined in the standard way, relative to an arbitrary set functor $\fun$, as follows:
\begin{definition}%
\label{d:supp}
Given $\xi \in \funT X$, a \emph{support} of $\xi$ is a subset $X' \subseteq X$
such that there is some $\xi' \in \funT X'$ with $\funT \iota (\xi') = \xi$,
where $\iota : X' \to X$ is the inclusion map.
\end{definition}
If we add to the signature $\Si_{\funM}$ of monotone modal logic a box and
diamond modality accessing the support, we obtain a signature which does admit
a disjunctive basis. This is connected with Proposition~\ref{p:wpp-db} and the fact that the functor $\underline{\funM}$ preserves weak pullbacks, which was proved in~\cite{enqv:comp17}.
The supported companion functor features prominently in our completeness proof
(with Seifan) for the monotone $\mu$-calculus~\cite{enqv:comp17}.

%%% Local Variables:
%%% mode: latex
%%% TeX-master: "main.tex"
%%% End:

%% file: sec-simulation.tex
\section{Disjunctive automata and simulation}%
\label{s:sim}

We now introduce disjunctive automata, which serve as a coalgebraic
generalization of non-deterministic automata for the modal $\mu$-calculus.
We refer to the sections~\ref{ss:aut} and~\ref{sec:games} for background on,
respectively, automata and their connection with $\mu$-calculi, and the
infinite games in which their semantics is formulated.

\begin{definition}
A $(\La,\Prop)$-automaton $\bbA = (A,\Th,\Om,\ai)$ is said to be
\emph{disjunctive}  if $\Th(c,a)$ is disjunctive, for all colors $c \in \funP\Prop$ and all states $a \in A$.

Given a disjunctive basis $\D$, we say that $\bbA$ is a $\D$-automaton if
$\Th(c,a) \in \D(A)$, for all colors $c \in \funP\Prop$ and all states $a \in A$.
\end{definition}

\begin{definition}
Let $\bbA$ be a $\La$-automaton and let $(\bbS,s_{I})$ be a pointed
$\funT$-model.
A strategy $f$ for $\eloi$ in $\AG(\bbA,\bbS)@(a,s)$ is \emph{dividing} if for
every $t$ in $\bbS$ there is at most one state $b$ in $\bbA$ such that the
position $(b,t)$ is \emph{$f$-reachable} (i.e., occurs in some $f$-guided
match).
We say that
$\bbA$ \emph{strongly accepts} $(\bbS,s_{I})$, notation:
$\bbS,s_{I} \sat_{s} \bbA$ if $\eloi$ has a dividing winning strategy
in the game $\AG(\bbA,\bbS)@(a_{I},s_{I})$.
\end{definition}

Disjunctive automata are very well behaved.
For instance, the following observation states a \emph{linear-size} model
property.
This observation goes back to Janin \& Walukiewicz~\cite{jani:auto95}, who
proved the linear-size model property for so-called disjunctive formulas of
the modal $\mu$-calculus.
Their result was generalised to the coalgebraic setting of $\nb$-based
automata by Kupke \& Venema~\cite{kupk:coal08}.

\begin{thm}
Let $\bbA = (A,\Th,a_I,\Om)$ be a disjunctive automaton for a set functor
$\funT$.
If $\bbA$ accepts some pointed $\funT$-model, then it accepts one of which the
carrier $S$ satisfies $S \sse A$.
\end{thm}

The main property of disjunctive automata, which we will use throughout the
remainder of this paper, is the following.

\begin{prop}%
\label{p:dj1}
Let $\bbA$ be a disjunctive $\La$-automaton.
Then any pointed $\funT$-model which is accepted by $\bbA$ has a pre-image
model which is strongly accepted by $\bbA$.
\end{prop}

\begin{proof}
Let $\bbS = (S,\si,V)$ be a pointed $\funT$-model, let $s_{I} \in S$, and let
$f$ be a winning strategy for $\eloi$ in the acceptance game $\AG \isdef
\AG(\bbA,\bbS)@(\ai,s_{I})$.
By Fact~\ref{f:pdpg} we may without loss of generality assume that $f$ is
positional.
We will construct (i) a pointed $\funT$-model $(X,\xi,W,x_{I})$,
(ii) a tree $(X,R)$ which is rooted at $x_{I}$ (in the sense that for every $t
\in X$ there is a unique $R$-path from $x_{I}$ to $x$) and supports $(X,\xi)$
(in the sense that $\xi(x) \in \funT R(x)$, for every $x \in X$), and
(iii) a morphism $h: (X,\xi,W) \to (S,\si,V)$ such that $h(x_{I}) = s_{I}$.
% , and a subset $U \sse T$ of so-called \emph{active} nodes.
In addition $(X,\xi,W,x_{I})$ will be strongly accepted by $\bbA$.

More in detail, we will construct all of the above step by step, and by a
simultaneous induction we will associate, with each
% active
$t \in X$ of depth $k$,
a (partial) $f$-guided match $\Si_{t}$ of length $2k +1$; we will denote the
final position of $\Si_{t}$ as $(a_{t},s_{t})$, and will define $h(t) \isdef
s_{t}$.

For the base step of the construction we take some fresh object $x_{I}$, we
define $\Si_{x_{I}}$ to be the match consisting of the single position $(\ai,
s_{I})$, and set $h(x_{I}) \isdef s_{I}$.

Inductively assume that we are dealing with a node $t \in X$ of depth $k$,
and that $\Si_{t}$, $a_{t}$ and $s_{t}$ are as described above.
Since $\Si_{t}$ is an $f$-guided match and $f$ is a winning strategy in $\AG$,
the pair $(a_{t},s_{t})$ is a winning position for $\eloi$ in $\AG$.
In particular, the marking $m_{t}: S \to \funP A$ prescribed by $f$ at this
position satisfies
\[
S,\si(s_{t}),m_{t} \satone \Th(a_{t},V^{\flat}(s_{t})).
\]
Now by disjunctiveness of the automaton $\bbA$ there is a set $R(t)$ (that we
may take to consist of fresh objects), an object $\xi(t) \in \funT R(t)$,
an $A$-marking $m'_{t}: R(t) \to \funP A$ and a map $h_{t}: R(t) \to S$, such
that\footnote{%
   To simplify our construction, we strengthen clause (3) in
   Definition~\ref{d:dj}.
   This is not without loss of generality, but we may take care of the general
   case using a routine extension of the present proof.}
$\sz{m(u)} = 1$ and $m'_{t}(u) \sse m_{t}(h_{t}(u))$ for all $u\in R(t)$,
$(\funT h_{t}) \xi(t) = \si(s_{t})$ and
\[
R(t),\xi(t),m'_{t} \satone \Th(a_{t},V^{\flat}(s_{t})).
\]
% Call a node $u \in R(t)$ \emph{active} if $m'_{t}(u) \neq \nada$, and if so let
Let $a_{u}$ be the unique object such that $m'_{t}(u) = \{ a_{u}\}$, define
$s_{u} \isdef h_{t}(u)$, and let $\Si_{u}$ be the match $\Si_{u} \isdef \Si_{t}
\cdot m_{t} \cdot (a_{u},s_{u})$.

With $(X,R,x_{I})$ the tree constructed in this way, and observing that $\xi(t)
\in R(t) \sse X$, we let $\xi$ be the coalgebra map on $X$.
Taking $h: X \to S$ to be the union $\{(x_{I},s_{I})\} \cup \{ h_{t} \mid t \in
X \}$, we can easily verify that $h$ is a surjective coalgebra morphism.
Finally, we define the valuation $W: \Prop \to \funP X$ by putting $W(p) \isdef
\{ x \in X \mid hx \in V(p) \}$.

It remains to show that $\bbA$ strongly accepts the pointed $\funT$-model
$(\bbX,x_{I})$, with $\bbX = (X,\xi,W)$; for this purpose consider the following
(positional) strategy $f'$ for $\eloi$ in $\AG(\bbA,\bbX)$.
At a position $(a,t) \in A \times X$ such that $a \neq a_{t}$ $\eloi$ moves
randomly (we may show that such a position will not occur);
on the other hand, at a position of the form $(a_{t},t)$, the move suggested by
the strategy $f'$ is the marking $m'_{t}$.
Then it is obvious that $f'$ is a dividing strategy; to see that $f'$ is
winning from starting position $(\ai,x_{I})$, consider an infinite match
$\Si$ of $\AG(\bbA,\bbX)@(\ai,x_{I})$ (finite matches are left to the reader).
It is not hard to see that $\Si$ must be of the form
$\Si = (a_{0},x_{0})m'_{x_{0}}(a_{1},x_{1})m'_{x_{1}}\cdots$, where
$\Si^{-} = (a_{0},h(x_{0}))m_{x_{0}}(a_{1},h(x_{1}))m_{x_{1}}\cdots$
is an $f$-guided match of $\AG$.
From this observation it is immediate that $\Si$ is won by $\eloi$.
\end{proof}

% \newpage

We now come to our main application of disjunctive bases, and fill in the main
missing piece in the theory of coalgebraic automata based on predicate liftings:
a simulation theorem.
As mentioned in the introduction, Janin \& Walukiewicz' simulation
theorem~\cite{jani:auto95} is one of the key tools in the theory of the standard
modal $\mu$-calculus, see for instance~\cite{vene:lect12} for many examples.
At the coalgebraic level of generality, a first simulation theorem was proved
by Kupke \& Venema~\cite{kupk:coal08} for $\nb$-based automata.

\begin{thm}[Simulation]
Let $\La$ be a monotone modal signature for the set functor $\funT$ and
assume that $\La$ has a disjunctive basis $\D$.
Then there is an effective construction transforming an arbitrary $\La$-automaton
$\bbA$ into an equivalent $\D$-automaton $\simof{\bbA}$.
\end{thm}

\begin{proof}
Let $\bbA = (A,\Th,\Om,
\ai)$ be a $\La$-automaton.
Our definition of $\simof{\bbA}$ is rather standard~\cite{vene:lect12},
so we will confine ourselves to the definitions.
The construction takes place in two steps, a `pre-simulation' step that produces
a disjunctive automaton $\mathsf{pre}(\bbA)$ with a non-parity acceptance
condition, and a second `synchronization' step that turns this nonstandard
disjunctive automaton into a standard one.
Both steps of the construction involve a `change of base' in the sense that
we obtain the transition map of the new automaton via a substitution relating
its carrier to the carrier of the old automaton.

We define the pre-simulation automaton of $\bbA$ as the structure
\[
    \mathsf{pre}(\bbA) \isdef (\shA,\sh{\Th},\NBT_{\bbA},R_{I}),
\]
where
the carrier of the pre-simulation $\mathsf{pre}(\bbA)$ of $\bbA$ is the
collection $\shA$ of \emph{binary} relations over $A$, and the initial state
$R_I$ is the singleton pair $\{(\ai,\ai)\}$.
For its transition function, first define the map $\Th^{\star}: A \times
\funP\Prop \to \MLoLa^{+}(A\times A)$ by putting, for $a \in A$ and
$c \in \funP\Prop$:
\[
\Th^{\star}(a,c) \isdef \Th(a,c)[\theta_{a}],
\]
where $\theta_{a}: A \to \Latt(A \times A)$ is the
\emph{tagging} substitution given by $\theta_{a}: b \mapsto (a,b)$.
Now, given a state $R \in \shA$ and color $c \in \funP\Prop$, take
$\sh{\Th}(R,c)$ to be an arbitrary but fixed formula in $\D(\shA)$ such that
\[
\sh{\Th}(R,c)[\land_{A\times A}] \equiv \bw_{a \in \Ran R}
\Th^{\star}(a,c).
\]
Clearly such a formula exists by our assumption on $\D$ being a disjunctive
basis for $\La$.

Turning to the \emph{acceptance condition},
define a \emph{trace} on an $\shA$-stream $\rho = {(R_{n})}_{0 \leq n<\om}$ to be
an $A$-stream $\al = {(a_{n})}_{0\leq n < \om}$ with $R_{i}a_{i}a_{i+1}$ for
all $i \leq 0$.
Calling such a trace $\al$ \emph{bad} if $\max\{ \Om(a) \mid a \text{ occurs
infinitely often in } \al \}$ is odd, we obtain the acceptance condition of
the automaton $\mathsf{pre}(\bbA)$ as the set $\NBT_{\bbA} \sse {(\shA)}^{\om}$
of $\shA$-streams that contain no bad trace.

Finally we produce the simulation of $\bbA$ by forming a certain kind of product
of $\mathsf{pre}(\bbA)$ with $\mathbb{Z}$, where $\mathbb{Z} = (Z,\delta,\Om',
z_I)$ is some deterministic parity stream automaton recognizing the
$\omega$-regular language $\NBT_{\bbA}$.
More precisely, we define $\simof{\bbA} \isdef (\shA \times Z,\Th'',\Om'',
(R_{I},z_I))$ where:
\begin{itemize}
    \item $\Th''(R,z) \isdef \sh{\Th}(R)[(Q,\delta(R,z)/Q) \mid Q \in \shA]$ and
    \item $\Om''(R,z) \isdef \Om'(z)$.
\end{itemize}
The equivalence of $\bbA$ and $\simof{\bbA}$ can be proved by
relatively standard means~\cite{vene:lect12}.
\end{proof}
%%% Local Variables:
%%% mode: latex
%%% TeX-master: "main.tex"
%%% End:

%% file: sec-lyndon.tex
\section{Lyndon theorems}%
\label{s:lyn}

Lyndon's classical theorem in model theory provides a syntactic characterization
of a semantic property, showing that a formula is \emph{monotone} in a predicate
$P$ if and only if it is equivalent to a formula in which $P$ occurs only
\emph{positively}.
A version of this result for the modal $\mu$-calculus was proved by d'Agostino
and Hollenberg in~\cite{dago:logi00}.
Here, we show that their result holds for any $\mu$-calculus based on a
signature that admits a disjunctive basis.

We first turn to the one-step version of the Lyndon Theorem, for which we need
the following definition; we also recall the substitutions $\land_{A}$ and
$\lor_{A}$ defined in Section~\ref{s:prel}.

\begin{definition}
A \emph{propositional $A$-type} is a subset of $A$.
For $B \subseteq A$ and $a \in A$, the formulas $\tau_B$ and $\tau_B^{a+}$ are
defined by:
\[\begin{array}{lll}
   \tau_B      & := & \bigwedge B \wedge
      \bigwedge \{\neg a \mid a \in A \setminus B\}
\\ \tau_B^{a+} & := &\bigwedge B \wedge
\bigwedge \{\neg b \mid b \in A \setminus (B\cup \{a\})\}
\end{array}\]
We let $\tau$ and $\tau^{a+}$ denote the maps $B \mapsto \tau_B$ and $B \mapsto
\tau_B^{a+}$, respectively.
\end{definition}

\begin{prop}%
\label{p:intermediate-normal-form}
Suppose $\La$ admits a disjunctive basis. Then for any formula  $\alpha$ in
$\MLoneLa(A)$ there is a one-step equivalent formula of the form
$\delta[\vee_{\funP A}][\tau]$ for some $\delta \in \D(\funP \funP A)$.
\end{prop}

\begin{proof}
We first check that everything is correctly typed: note that we have
$\vee_{\funP A} : \funP\funP A \to \Bool(\funP A)$ and so $\delta[\vee_{\funP A}] \in
\MLoneLa(\funP A)$, and $\tau_{\funP A} : \funP{A} \to \Bool (A)$.
So $\delta[\vee_{\funP A}][\tau] \in \MLoneLa(A)$, as required.

For the normal form proof, first note that we can use boolean duals of the
modal operators to push negations down to the zero-step level.
Putting the resulting formula in disjunctive normal form, we obtain a
disjunction of formulas of the form
$\hs_{\lambda_1} \ol{\pi_1} \wedge \dots \wedge \hs_{\lambda_k} \ol{\pi_k}$,
where all the $\pi$-formulas are in $\Bool(A)$.
Repeatedly applying the distributivity of $\D$ over $\Lambda$ and the
distributive law for $\D$, we can rewrite each such disjunct as a formula of the
form $\delta[\sigma]$ where, for some set $B$, $\delta \in \D(B)$ and
$\sigma: B \to \Bool(A)$ is some propositional substitution.
Now, just apply propositional logic to rewrite each formula $\si_{b}$ as a
disjunction of formulas in $\tau[\funP A]$, and we are done.
\end{proof}

\begin{thm}[One-step Lyndon theorem]%
\label{t:onestep-lyndon}
Let $\La$ be a monotone modal signature for the set functor $\funT$ and
assume that $\La$ has a disjunctive basis.
Any $\alpha \in \MLoneLa(A)$, monotone in the variable $a \in A$,
is one-step equivalent to some formula in $\MLoneLa(A)$, which is positive in $a$.
\end{thm}

\begin{proof}
By Proposition~\ref{p:intermediate-normal-form}, we can assume that $\alpha$ is
of the form  $\delta[\vee_{\funP A}][\tau]$ for some $\delta \in
\D(\funP A)$.
Clearly it suffices to show that:
\[
    \delta[\vee_{\funP A}][\tau] \oneseq \delta[\vee_{\funP A}][\tau^{a+}]
\]
One direction, from left to right, is easy since $\delta[\vee_{\funP A}]$ is a
monotone formula in $\MLoneLa(\funP A)$, and $\mngzero{\tau_B}_m \subseteq
\mngzero{\tau_B^{a+}}_m$ for each $B \subseteq A$ and each marking $m : X \to
\funP A$.

For the converse direction, suppose $X,\xi,m \satone
\delta[\vee_{\funP A}][\tau^{a+}]$.
We define a $\funP A$-marking $m_0 : X \to \funP \funP A$ by setting $m_0(u) \isdef
\{B \subseteq A \mid B \preceq_a m(u)\}$, where the relation $\preceq_a$ over
$\funP A$ is defined by $B \preceq_a B'$ iff $B \setminus \{a\} = B' \setminus
\{a\}$, and $a \notin B$ or $a \in B'$.
We claim that $X,\xi,m_0 \satone  \delta[\vee_{\funP A}]$.
Since $\delta[\vee_{\funP A}]$ is a monotone formula, it suffices to check that
$\mngzero{\tau_B^{a+}}_m \subseteq \mngzero{B}_{m_0}$ for each $B \subseteq A$.
This follows by just unfolding definitions.

Since $\delta$ was disjunctive, so is $\delta[\vee_{\funP A}]$, as an easy
argument will reveal.
So we now find a one-step frame morphism $f : (X',\xi') \to (X,\xi)$,
together with a marking $m' : X' \to \funP \funP A$ such that $\sz{m'(u)}
\leq 1$ and $m'(u) \subseteq m_0(f(u))$ for all $u \in X'$, and such that
$X',\xi',m' \satone \delta[\vee_{\funP A}]$.
We define a new $A$-marking $m'' : X' \to \funP A$ on $X'$ by setting $m''(u) = B$,
if $m'(u) = \{B\}$, and $m''(u) = m (f(u))$ if $m'(u) = \emptyset$.
Note that, for each $B \subseteq A$, we have $\mngzero{B}_{m'} \subseteq
\mngzero{\tau_B}_{m''}$, so by monotonicity of $\delta[\vee_{\funP A}]$ we get
$X',\xi',m'' \satone \delta[\vee_{\funP A}][\tau]$.

Comparing the markings $m''$ and $m \circ f$, we claim that $m''(u) \preceq_a
m(f(u))$ for all $u \in X'$.
If $m'(u) = \emptyset$, then in fact  $m''(u) = m(f(u))$ by definition of $m''$.
If $m'(u) = \{B\}$, then $m''(u) = B \in m'(u) \subseteq m_0(f(u))$, hence $B
\preceq_a m(f(u))$ by definition of $m_0$.
Since $\delta[\vee_{\funP A}][\tau]$ was monotone with respect to the variable
$a$ it follows  that $X',\xi',m \circ f \satone \delta[\vee_{\funP A}][\tau]$
and so $X,\xi,m \satone \delta[\vee_{\funP A}][\tau]$ by naturality, thus
completing the proof of the theorem.
\end{proof}

A useful corollary to this theorem is that, given an expressively complete set
$\Lambda$ of predicate liftings for a functor $\funT$, the language
$\muML_\Lambda$ has the same expressive power as the full language $\muML_\funT$.
At first glance this proposition may seem trivial, but it is important to see
that it is not: given a formula $\varphi$ of $\muML_\funT$, a naive definition
of an equivalent formula in $\muML_\Lambda$ would be to apply expressive
completeness to simply replace each subformula of the form
$\hs_\lambda(\psi_1,\dots,\psi_n)$ with an equivalent one-step formula $\alpha$
over $\{\psi_1,\dots,\psi_n\}$, using only predicate liftings in $\Lambda$.
But if this subformula contains bound fixpoint variables, these must still appear
positively in $\alpha$ in order for the translation to even produce a
grammatically correct formula!
We need the stronger condition of \emph{Lyndon completeness} for $\Lambda$.
Generally, we have no guarantee that expressive completeness entails Lyndon
completeness, but in the presence of a disjunctive basis, we do: this is a
consequence of Theorem~\ref{t:onestep-lyndon}.

\begin{cor}%
\label{c:lyn}
Suppose $\Lambda$ is an expressively complete set of monotone predicate liftings
for $\funT$.
If $\Lambda$ admits a disjunctive basis, then $\Lambda$ is Lyndon complete and hence $\muML_\Lambda \equiv \muML_\funT$.
\end{cor}

\begin{proof}
The simplest proof uses automata: pick a modal $\Lambda'$-automaton $\bbA$,
where $\Lambda'$ is the set of all monotone predicate liftings for $\funT$, and
apply expressive completeness to replace each formula $\alpha$ in the co-domain
of the transition map $\Theta$ with an equivalent one-step formula $\alpha'$
using only liftings in $\Lambda$.
This formula is still monotone in all the variables in $A$ since it is
equivalent to $\alpha$, so we can apply the one-step Lyndon
Theorem~\ref{t:onestep-lyndon} to replace $\alpha'$ by an equivalent and positive
one-step formula $\beta$ in $\MLoneLa(A)$.
Clearly, the resulting automaton $\bbA'$ will be semantically equivalent to
$\bbA$.
\end{proof}

We now turn to our Lyndon Theorems for the full coalgebraic modal (fixpoint)
languages.
Let ${(\muMLLa)}^{M}_{p}$ and ${(\MLLa)}^{M}_{p}$ denote the fragments of
respectively $\muMLLa$ and $\MLLa$, consisting of the formulas that are positive
in the proposition letter $p$.

\begin{thm}[Lyndon Theorem]%
\label{t:lyn}
There is an effective translation ${(\cdot)}^{M}_{p}: \muMLLa \to {(\muMLLa)}^{M}_{p}$,
which restricts to a map ${(\cdot)}^{M}_{p}: \MLLa \to {(\MLLa)}^{M}_{p}$, and
satisfies that
\[
\phi \in \muML \text{ is monotone in } p \text{ iff } \phi \equiv
\phi^{M}_{p}.
\]
\end{thm}

\begin{proof}
By the equivalence between formulas and $\La$-automata and the Simulation
Theorem, it suffices to prove the analogous statement for disjunctive
coalgebra automata.

Given a disjunctive $\La$-automaton $\bbA = (A,\Th,\Om,\ai)$, we define
$\bbA^{M}_{p}$ to be the automaton $(A,\Th^{M}_{p},\Om,\ai)$, where
\[
\Th^{M}_{p}(a,c) \isdef
  \begin{cases}
  \Th(a,c) \vee \Th(a,c \setminus \{p\})& \text{if } p \in c
  \\ \Th(a,c) & \text{if } p \not\in c.
  \end{cases}
\]
Clearly $\bbA^{M}_{p}$ is a disjunctive automaton as well, and it is routine to
show that $\bbA^{M}_{p}$ is equivalent to a formula in $\muML_\La$ that is
positive in the variable $p$. The key observation is that we have the following valid equivalence of classical propositional logic, for any formulas $\pi,\Theta_1, \Theta_2$:
\[
    (\pi \wedge \neg p \wedge \Theta_1) \vee ( \pi \wedge p \wedge (\Theta_1 \vee  \Theta_2)) \quad \Leftrightarrow \quad (\pi \wedge  \Theta_1) \vee (\pi \wedge p \wedge \Theta_2)
\]
We claim that $\bbA$ is monotone in $p$ iff $\bbA \equiv \bbA^{M}_{p}$.
Leaving the easy direction from right to left to the reader, we prove the
opposite implication.
So assume that $\bbA$ is monotone in $p$.
Since it immediately follows from the definitions that $\bbA$ always implies
$\bbA^{M}_{p}$, we are left to show that $\bbA^{M}_{p}$ implies $\bbA$, and
since $\bbA^{M}_{p}$ is disjunctive, by  Proposition~\ref{p:dj1} and invariance
of acceptance by coalgebra automata it suffices to prove the following:
\begin{equation}
\label{eq:lyn1}
\bbS,s_{I} \sat_{s} \bbA^{M}_{p} \text{ implies } \bbS,s_{I} \sat \bbA,
\end{equation}
for an arbitrary $\funT$-model $(\bbS,s_{I})$.

To prove~\eqref{eq:lyn1}, let $f$ be a dividing winning strategy for $\eloi$
in $\AG^{M} \isdef \AG(\bbA^{M}_{p},\bbS)@(\ai,s_{I})$.
Our aim is to find a subset $U \sse V(p)$ such that $\bbS[p \mapsto U], s_{I}
\sat \bbA$; it then follows by mononotonicity that $\bbS,s_{I}\sat \bbA$.
Call a point $s \in S$ \emph{$f$-accessible} if there is a (by assumption unique)
state $a_{s}$ such that the position $(a_{s},s)$ is $f$-reachable in $\AG^{M}$.
We define $U$ as the set of $f$-accessible elements $s$ of $V(p)$ such that:
\[
    S, \si(s),m_{s} \satone \Th(a_{s},V^{\flat}(s)),
\]
where $m_{s}$ is the $A$-marking provided by $f$ at position $(a_{s},s)$.
We let
$V_{U}$ abbreviate $V[p \mapsto U]$.
We claim that
\begin{equation}
\label{eq:lyn2}
\text{if $s$ is $f$-accessible then } S, \si(s),m_{s} \satone
\Th(a_{s},V_{U}^{\flat}(s)).
\end{equation}
To see why~\eqref{eq:lyn2} holds, note that for any $f$-accessible point $s$,
the marking $m_{s}$ is a legitimate move at position $(a_{s},s)$, since $f$
is assumed to be winning for $\eloi$ in $\AG^{M}$.
In other words, we have $S, \si(s),m_{s} \satone \Th^M_p(a_{s},V^{\flat}(s))$.
We have to make a case distinction, between the following three cases:
\begin{description}
\item[Case 1] $p \notin V_U^{\flat}(s)$ and $p \notin V^{\flat}(s)$. Then $\Th(a_{s},V_U^{\flat}(s)) = \Th^M_p(a_{s},V^{\flat}(s))$, so we get $ S, \si(s),m_{s} \satone
\Th(a_{s},V_{U}^{\flat}(s))$ as required.
\item[Case 2] $p \notin V_U^{\flat}(s)$ but $p \in V^{\flat}(s)$. Then $
\Th(a_{s},V^{\flat}(s))$ does not hold in the one-step model $S, \si(s),m_{s}$, so by definition of $\Th^M_p$ we have $S, \si(s),m_{s} \satone
\Th(a_{s},V^{\flat}(s) \setminus \{p\})$. But  $V^{\flat}(s) \setminus \{p\} = V_U^{\flat}(s)$, so again we get $ S, \si(s),m_{s} \satone
\Th(a_{s},V_{U}^{\flat}(s))$ as required.
\item[Case 3] $p \in V_U^{\flat}(s)$. In this case we have $ S, \si(s),m_{s} \satone
\Th(a_{s},V^{\flat}(s))$ by definition of the valuation $V_U$. Furthermore, since $V^{\flat}(s) = V_U^{\flat}(s)$, we get  $ S, \si(s),m_{s} \satone
\Th(a_{s},V_{U}^{\flat}(s))$.
\end{description}
Hence,~\eqref{eq:lyn2} holds in all three possible cases.
Finally, it is straightforward to derive from~\eqref{eq:lyn2} that $f$ itself is
a (dividing) winning strategy for $\eloi$ in the acceptance game $\AG(\bbA,
\bbS)$ initialized at $(\ai,s_{I})$.
\end{proof}

\begin{remark}
Observe that as a corollary of Theorem~\ref{t:lyn} and the decidability of the
satisfiability problem of $\muMLLa$~\cite{cirs:expt09}, it is decidable whether
a given formula $\phi \in \muMLLa$ is monotone in $p$.
\end{remark}

%%% Local Variables:
%%% mode: latex
%%% TeX-master: "main.tex"
%%% End:

%% file: sec-unifitp.tex
\section{Uniform Interpolation}%
\label{s:ui}

Uniform interpolation is a very strong form of the interpolation theorem, first
proved for the modal $\mu$-calculus in~\cite{dago:logi00}.
It was later generalized to coalgebraic modal logics in~\cite{mart:unif15}.
However, the proof crucially relies on non-deterministic automata, and for that
reason the generalization in~\cite{mart:unif15} is stated for nabla-based
languages.
With a simulation theorem for predicate liftings based automata in place, we can
prove the uniform interpolation theorem for a large class of $\mu$-calculi based
on predicate liftings.

\begin{definition}
Given a formula $\phi \in \muMLLa$, we let $\Prop_{\phi}$ denote the set of
proposition letters occurring in $\phi$.
Given a set $\Prop$ of proposition letters and a single proposition letter $p$,
it may be convenient to denote the set $\Prop \cup \{ p \}$ as $\Prop p$.
\end{definition}

\begin{definition}
A logic $\mathcal{L}$ with semantic consequence relation $\models$ is said to
have the property of \emph{uniform interpolation} if, for any formula $\phi \in
\mathcal{L}$ and any set $\Prop \sse \Prop_{\phi}$ of proposition
letters, there is a formula $\phi_{\Prop} \in \mathcal{L}(\Prop)$,
effectively constructible from $\phi$, such that
\begin{equation}
\label{eq:ui}
\phi \models \psi \text{ iff } \phi_{\Prop} \models \psi,
\end{equation}
for every formula $\psi \in \mathcal{L}$ such that $\Prop_{\phi} \cap
\Prop_{\psi} \sse \Prop$.
\end{definition}

To see why this property is called uniform \emph{interpolation}, it is not hard
to prove that, if $\phi \models \psi$, with $\Prop_{\phi} \cap \Prop_{\psi}
\sse \Prop$, then the formula $\phi_{\Prop}$ is indeed an interpolant in the
sense that $\phi \models \phi_{\Prop} \models \psi$ and
$\Prop_{\phi_{\Prop}} \sse \Prop_{\phi} \cap \Prop_{\psi}$.

\begin{thm}[Uniform Interpolation]%
\label{t:ui}
Let $\La$ be a monotone modal signature for the set functor $\funT$ and
assume that $\La$ has a disjunctive basis.
Then both logics $\MLLa$ and $\muMLLa$ enjoy the property of uniform
interpolation.
\end{thm}

Following D'Agostino \& Hollenberg~\cite{dago:logi00}, we prove
Theorem~\ref{t:ui} by automata-theoretic means.
The key proposition in our proof is Proposition~\ref{p:ui1} below, which refers
to the following construction on disjunctive automata.

\begin{definition}
Let $\Prop$ be a set of proposition letters not containing the letter $p$.
Given a disjunctive $(\La,\Prop p)$-automaton $\bbA = (A,\Th,\Om,\ai)$, we
define the map $\Th^{\exists p}: A \times \funP\Prop \to \D(A)$ by
\[
\Th^{{\exists} p}(a,c) \isdef \Th(a,c) \lor \Th(a,c\cup \{ p \}),
\]
and we let $\bbA^{{\exists} p}$ denote the $(\La,\Prop)$-automaton
$(A, \Th^{{\exists} p},\Om,\ai)$.
\end{definition}

The following proposition shows that the operation ${(\cdot)}^{{\exists} p}$
behaves like an \emph{existential quantifier}, but with a twist: the automaton
$\bbA^{{\exists} p}$ accepts a pointed coalgebra model $(\bbS,s_{I})$ iff
for some subset $P$ of \emph{some preimage model}$(\bbS',s'_{I})$, the
model $(\bbS'[p \mapsto P],s'_{I})$ is accepted by $\bbA$.

\begin{prop}%
\label{p:ui1}
Let $\Prop \sse \PropQ$ be sets of proposition letters, both not containing the
letter $p$.
Then for any disjunctive $(\La,\Prop p)$-automaton $\bbA$ and any pointed
$\funT$-model $(\bbS,s_{I})$ over $\PropQ$:
\begin{equation}
\bbS,s_{I} \sat \bbA^{{\exists} p} \text{ iff }
\bbS',s'_{I} \sat_{s} \bbA \text{ for some $\PropQ p$-model $(S',s'_{I})$ such
that } \bbS'\rst{\PropQ},s'_{I} \simu \bbS,s_{I}.
\end{equation}
\end{prop}

\begin{proof}
We only prove the direction from left to right, leaving the other (easier)
direction as an exercise to the reader.
For notational convenience we assume that $\Prop = \PropQ$.

By Proposition~\ref{p:dj1} it suffices to assume that $(\bbS,s_{I})$ is
\emph{strongly} accepted by $\bbA^{{\exists} p}$ and find a subset $U$ of
$S$ for which we can prove that $\bbS[p \mapsto U],s_{I} \sat_{s} \bbA$.
So let $f$ be a dividing winning strategy for $\eloi$ in
$\AG(\bbA^{{\exists} p},\bbS)@(\ai,s_{I})$ witnessing that $\bbS,s_{I} \sat_{s}
\bbA^{{\exists} p}$.
Call a point $s \in S$ \emph{$f$-accessible} if there is a state $a \in A$ such
that the position $(a,s)$ is $f$-reachable; since this state is unique by the
assumption of strong acceptance we may denote it as $a_{s}$.
Clearly any position of the form $(a_{s},s)$ is winning for $\eloi$, and hence
by legitimacy of $f$ it holds in particular that
\[
S,\si(s),m_{s} \satone \Th^{{\exists} p}(a_{s}, V^{\flat}(s)),
\]
where $m_{s}: S \to \funP A$ denotes the marking selected by $f$ at position
$(a_{s},s)$.
Recalling that $\Th^{{\exists} p}(a_{s}, V^{\flat}(s)) =
\Th(a_{s}, V^{\flat}(s)) \lor \Th(a_{s}, V^{\flat}(s) \cup \{ p \})$, we define
\[
U \isdef \{ s \in S \mid s \text{ is $f$-accessible and }
S,\si(s),m_{s} \not\satone \Th(a_{s}, V^{\flat}(s)) \}.
\]
By this we ensure that, for all $f$-accessible points $s$:
\begin{align}
s \not\in U
  & \text{ implies } S,\si(s),m_{s} \satone \Th(a_{s}, V^{\flat}(s))
\label{eq:ui1}
\\ \text{ while } s \in U
  & \text{ implies } S,\si(s),m_{s} \satone \Th(a_{s}, V^{\flat}(s) \cup \{ p \})
\label{eq:ui2}
\end{align}
Now consider the valuation $V_{U} \isdef V[p \mapsto U]$, and observe that by
this definition we have $V_{U}^{\flat}(s) = V^{\flat}(s)$ if $s \not\in U$ while
$V_{U}^{\flat}(s) = V^{\flat}(s) \cup \{ p \}$ if $s \in U$.
Combining this with~\eqref{eq:ui1} and~\eqref{eq:ui2} we find that
\[
S, \si(s), m_{s} \satone \Th(a_{s}, V_{U}^{\flat})
\]
whenever $s$ is $f$-accessible.
In other words, $f$ provides a legitimate move $m_{s}$ in
$\AG(\bbA,\bbS)@(a_{s},s)$ at any position of the form $(a_{s},s)$.
From this it is straightforward to derive that $f$ itself is a (dividing)
winning strategy for $\eloi$ in $\AG(\bbA,\bbS[p \mapsto U])@(\ai,s_{I})$,
and so we obtain that $\bbS[p \mapsto U],s_{I} \sat_{s} \bbA$ as required.
\end{proof}

The remaining part of the argument follows by a fairly standard argument going
back to D'Agostino \& Hollenberg~\cite{dago:logi00} (see also Marti et
alii~\cite{mart:unif15}) --- a minor difference being that our quantification
operation ${(\cdot)}^{\exists p}$ refers to pre-images rather than to bisimilar
models.

\begin{prop}%
\label{p:bq}
Given any proposition letter $p$, there is a map $\widetilde{\exists} p$ on
$\muMLLa$, restricting to $\MLLa$, such that $\Prop_{\widetilde{\exists} p. \phi}
= \Prop_{\phi} \setminus \{ p \}$ and, for every pointed $(\bbS,s_{I})$ over a
set $\PropQ \supseteq \Prop_{\phi}$ with $p \not\in \PropQ$:
\begin{equation}
\label{eq:bq}
\bbS,s_{I} \sat \widetilde{\exists} p. \phi \text{ iff }
\bbS',s'_{I} \sat \phi \text{ for some $\PropQ p$-model $(S',s'_{I})$ such
that } \bbS'\rst{\PropQ},s'_{I} \simu \bbS,s_{I}.
\end{equation}
\end{prop}

\begin{proof}
Straightforward by the equivalence between formulas and $\La$-automata, the
Simulation Theorem, and Proposition~\ref{p:ui1}.
\end{proof}

\noindent
\begin{proofof}{Theorem~\ref{t:ui}}
With $p_{1},\ldots,p_{n}$ enumerating the proposition letters in
$\Prop_{\phi} \setminus \Prop$, set
\[
\phi_{\Prop} \isdef
\widetilde{\exists} p_{1} \widetilde{\exists} p_{2} \cdots \widetilde{\exists}
p_{n}. \phi.
\]
Then a relatively routine exercise shows that $\phi \models \psi$ iff
$\phi_{\Prop} \models \psi$, for all formulas $\psi \in \muMLLa$ such that
$\Prop_{\phi} \cap \Prop_{\psi} \sse \Prop$.
For some detail, first assume that $\phi \models \psi$, and take an arbitrary
pointed model $(\bbS_{0},s_{0})$ over some set $\PropQ \supseteq \Prop_{\psi}
\cup \Prop$ such that $\PropQ \cap \{ p_{1},\ldots, p_{n} \} = \nada$ and
$\bbS_{0},s_{0} \sat \phi_{\Prop}$.
Then successive applications of Proposition~\ref{p:bq} provide,
for $i = 1, \ldots, n$, pointed $\funT$-models $(\bbS_{i},s_{i})$ over
$\PropQ \cup \{ p_{1}, \ldots, p_{i} \}$ such that
$(\bbS_{i},s_{i}) \sat
\widetilde{\exists} p_{i+1} \cdots \widetilde{\exists} p_{n}. \phi$ and
$\bbS_{i+1}\rst{(\PropQ \cup \{p_{1},\ldots,p_{i}\})},s_{i+1} \simu
\bbS_{i},s_{i}$, for all $i$.
Thus in particular we have $\bbS_{n},s_{n} \sat \phi$, from which it follows
by assumption that $\bbS_{n},s_{n} \sat \psi$, and since
$\bbS_{n}\rst{\PropQ},s_{n} \simu \bbS_{0},s_{0}$ by transitivity of $\simu$,
this implies that $\bbS_{0},s_{0} \sat \psi$, as required.
Conversely, to show that $\phi_{\Prop} \models \psi$ implies $\phi \models
\psi$, it suffices to prove that $\phi \models \phi_{\Prop}$; but this is
almost immediate from the definitions.

Finally, it is not difficult to verify that $\phi_{\Prop}$ is fixpoint-free if
$\phi$ is so; that is, the uniform interpolants of a formula in $\MLLa$ also
belong to $\MLLa$.
\end{proofof}

%%% Local Variables:
%%% mode: latex
%%% TeX-master: "main.tex"
%%% End:

%% file: sec-combined.tex
\section{Disjunctive bases for combined modal logics}

An important topic in modal logic concerns methods to construct complex logics
from simpler ones, in such a way that desirable metalogical properties transfer
to a combined logic from its components.
For an overview of this field, see~\cite{kuru:comb06}.
This thread has also been picked up in coalgebraic modal logic: in a particularly
interesting paper~\cite{cirs:modu07}, C\^{\i}rstea \& Pattinson provide generic
methods to obtain, from modal signatures
% $\La_{1}$ and $\La_{2}$
for two set functors $\funT_{1}$ and $\funT_{2}$, modal signatures for the
coproduct, product, and composition of $\funT_{1}$ and $\funT_{2}$.
Since the verification that a particular modal logic has a disjunctive basis can
be quite non-trivial --- as witnessed here by the case of graded modal logic ---
it will be useful to have some methods available that guarantee the existence of
a disjunctive basis being preserved by C\^{\i}rstea \& Pattinson's modular
constructions of signatures.
In fact, we will show that disjunctive bases can be constructed in a modular
fashion as well: in each of the cases of coproduct, product and composition we
will give an explicit construction of a disjunctive basis for the combined
signature, based on disjunctive bases for the composing functors.

\subsection{Coproduct}

Let $\funT = \funT_{1} + \funT_{2}$ be the coproduct of the set functors
$\funT_{1}$ and $\funT_{2}$.
For $i = 1,2$ we will use $\kappa_{i}$ to denote the natural
transformation $\kappa_{i}: \funT_{i} \Rightarrow \funT$ that instantiates to
the coproduct insertion map ${(\ka_{i})}_{S}: \funT_{i}S \to \funT S$, for every
set $S$.
Now suppose that we have been given signatures for $\funT_{1}$ and $\funT_{2}$,
respectively.
Following C\^{\i}rstea \& Pattinson~\cite{cirs:modu07}, we define the
combined signature $\La_{1} \oplus \La_{2}$ for $\funT$ as follows.

\begin{definition}
Where $\La_{1}$ and $\La_{2}$ are monotone modal signatures for $\funT_{1}$ and
$\funT_{2}$ respectively, we define
\[
\La_{1} \oplus \La_{2} \isdef \{ \kappa_{i}\circ\la \mid \la \in \La_{i} \}.
\]
In the syntax we shall write $\nxt{i,\la}$ rather than
$\mop{\kappa_{i}\circ\la}$.
\end{definition}

We leave it for the reader to check that $\La_{1} \oplus \La_{2}$ is indeed a
collection of monotone predicate liftings for $\funT$.
The meaning of the $\La_{1} \oplus \La_{2}$-modalities in an arbitrary
$\funT$-coalgebra $\bbS$ is given as follows:
\begin{equation}
\label{eq:cpr1}
\bbS, s \sat \nxt{i,\la}(\phi_{1},\ldots,\phi_{n}) \text{ iff }
  \begin{array}{l}
  \si(s) = {(\kappa_{i})}_{S}(\si^{i}) \text{ for some } \si^{i} \in \funT_{i}S \\
  \text{with } \si^{i} \in \la_{S}(\mng{\phi_{1}}^{\bbS},\ldots,\mng{\phi_{n}}^{\bbS}).
  \end{array}
\end{equation}
Note that the subscript $i$ in $\nxt{i,\la}$ works as a \emph{tag} indicating to
which part of the coproduct $\funT S = \funT_{1}S + \funT_{2}S$ the unfolding
$\si(s)$ of the state $s$ is situated.

The result on the existence of disjunctive bases that we want to prove is the
following.
\begin{thm}%
\label{t:copr}
Suppose both signatures $\Lambda_1$ and $\Lambda_2$ admit a disjunctive basis.
Then so does $\Lambda_1 \oplus \Lambda_2$.
\end{thm}

We start with giving a disjunctive basis for the combined signature.

\begin{definition}
Fix a set $A = \{ a_{1}, \ldots, a_{n} \}$.
Given a one-step formula $\al$ in the language $\MLone_{\La_{i}}(A)$, we let
$\al^{i} \in \MLone_{\Lambda_1 \oplus \Lambda_2}(A)$ denote the result of
replacing every occurrence of a modality $\hs_{\la}$ with $\nxt{i,\la}$.
Define
\[
(\D_{1}\oplus\D_{2})(A) \isdef
\{ \de_{1}^{1} \lor \de_{2}^{2} \in \MLone_{\Lambda_1 \oplus \Lambda_2}(A)
\mid \de_{i} \in \D_{i}(A) \}.
\]
where $\D_{1}$ and $\D_{2}$ are disjunctive bases for $\La_{1}$ and $\La_{2}$,
respectively.
\end{definition}

It remains to show that $\D_{1}\oplus\D_{2}$ is a disjunctive basis for the
signature $\Lambda_1 \oplus \Lambda_2$.

\begin{proofof}{Theorem~\ref{t:copr}}
We first show that $\D_{1}\oplus\D_{2}(A)$ consists of disjunctive formulas.
For this purpose we fix a set $A$ and an arbitrary formula in
$\D_{1}\oplus\D_{2}(A)$, say, $\de^{1}$ such that $\de \in \D_{1}(A)$.
Let $(S,\si,m)$ be an arbitrary one-step $\funT$-model such that $S,\si,m
\satone \de^{1}$, and make a case distinction.
If $\si = {(\kappa_{2})}_{S}(\si^{2})$ for some $\si^{2} \in \funT_{2}S$ then
% it is easy to see that $\de^{1}$ cannot contain any modalities. But then
a straightforward induction will show that $S,\si,m' \satone \de^{1}$,
where $m'(s) \isdef \nada$ for every $s \in S$.

If, on the other hand, $\si = {(\kappa_{1})}_{S}(\si^{1})$ for some $\si^{1} \in
\funT_{1}S$ then a routine inductive proof will reveal that the one-step
$\funT_{1}$-model $(S,\si^{1},m)$ satisfies $\de$.
By disjunctiveness of $\D_{1}$ we then obtain a separating cover for
$(S,\si^{1},m)$, consisting of a one-step $\funT_{1}$-model $(S',\si',m')$ and
a map $f: S' \to S$.
It is then easy to verify that the one-step $\funT$-model $(S',\kappa_{1}(\si'),
m')$, together with the same map $f$, is a separating cover for the one-step
$\funT$-model $(S,\si,m)$.
\medskip

It is left to check that $\D_{1}\oplus\D_{2}(A)$ satisfies the closure
conditions of disjunctive bases.
We leave it to the reader to verify that (modulo equivalence) the set
$\D_{1}\oplus\D_{2}(A)$ is closed under taking disjunctions, and contains the
formula $\top$.
For condition (2), take a formula of the form $\nxt{i,\la}\ol{\pi}$, where
$\la$ is a predicate lifting in $\La_{i}$; without loss of generality assume
$i = 1$.
Then by assumption there is a formula $\de \in \D_{1}(\funP A)$ such that
$\hs_{\la}\ol{\pi} \equiv^{1} \de[\wedge_{A}]$.
It is straightforward to verify that $\al \equiv^{1} \be$ implies $\al^{j}
\equiv^{1} \be^{j}$, for any pair of formulas $\al,\be \in
\MLone_{\La_{j}}^{+}(A)$.
But then it is immediate that $\nxt{1,\la}\ol{\pi} \equiv^{1} \de^{1}[\wedge_{A}]
\equiv^{1} (\de^{1} \lor \bot)[\wedge_{A}]$; clearly this suffices, since $\de^{1}
\lor\bot \in \D_{1}\oplus\D_{2}(\funP A)$.

Finally, for condition (3), consider the conjunction of two formulas in the sets
$\D_{1}\oplus\D_{2}(A)$ and $\D_{1}\oplus\D_{2}(B)$, respectively.
Using the distributive law of conjunctions over disjunctions, we may rewrite
this conjunction into an equivalent disjunction of formulas of the form $\al^{i}
\land \be^{j}$,
% $\al^{i} \in \D_{1}\oplus\D_{2}(A)$ and $\be^{j} \in \D_{1}\oplus\D_{2}(B)$,
where $\al \in \D_{i}(A)$ and $\be \in \D_{j}(B)$ for some $ \{ i,j \} \sse
\{ 1,2 \}$.
Clearly then it suffices to show that each conjunction of the latter form can be
rewritten into the required shape.
We distinguish two cases.

If $i = j$, then since $\D_{i}$ is a disjunctive basis for $\La_{i}$, there is
a formula $\ga^i \in \D_{i}(A \utimes B)$ such that $\al \land \be \equiv^{1}
\ga^i[\theta_{A,B}]$.
It is then straightforward to verify that $\al^{i} \land \be^{i} \equiv^{1}
\ga^{i}[\theta_{A,B}] \equiv^{1} (\ga^{i} \lor \bot)[\theta_{A,B}]$.

If, on the other hand, $i$ and $j$ are distinct, then we may without loss of
generality assume that $i=1$ and $j=2$.
We claim that in fact for \emph{any} pair of formulas $\al \in
\MLone^{+}_{\La_{1}}(A)$ and $\be \in \MLone^{+}_{\La_{2}}(B)$ (i.e, we do not need to
assume that $\al \in \D_{1}(A)$ or $\be \in \D_{2}(B)$), the conjunction
$\al^{1} \land \be^{2}$ is equivalent to a formula from the set $\{ \top, \bot,
\al^{1}, \be^{2} \}$.
The key observation in the proof of this claim is that if $\al =
\hs_{\la}\ol{\pi}$ and $\be = \hs_{\eta}\ol{\rho}$, then $\al^{1} \land \be^{2}
\equiv \bot$ --- this easily follows from the observation about tagging that we
made just after~\eqref{eq:cpr1}.
Finally, note that $\{ \top, \bot, \al^{1}, \be^{2} \} \sse \D_{1}\oplus\D_{2}
(A \cup B)$.
This means that every formula $\ga$ in this set belongs to $\D_{1}\oplus\D_{2}
(A \utimes B)$, and in addition satisfies that $\ga[\theta_{A,B}] = \ga$.
But then clearly the claim suffices to find a formula $\ga \in
\D_{1}\oplus\D_{2}(A \utimes B)$ such that $\al^{1} \land \be^{2} \equiv^{1}
\ga[\theta_{A,B}]$, as required.
\end{proofof}

\subsection{Product}

Given two functors $\funT_1,\funT_2$  and modal signatures $\Lambda_1,\Lambda_2$
for these functors respectively, we construct a new modal signature
$\Lambda_1 \otimes \Lambda_2 $ that contains the modalities of both $\Lambda_1$ and $\Lambda_2$. We want to interpret this combined language on coalgebras that can be seen simultaneously as both $\funT_1$- and $\funT_2$-coalgebras, and the natural choice is to form the product of the two functors and consider $\funT_1 \times \funT_2$-coalgebras.
\begin{definition}
Suppose  $\Lambda_1,\Lambda_2$ are modal signatures for functors $\funT_1,\funT_2$ respectively. Then the modal signature $\Lambda_1 \otimes \Lambda_2$ for $\funT_1 \times \funT_2$ is defined by:
\[
    =\Lambda_1 \otimes \Lambda_2 = \{ \funQ \pi_i \circ \lambda \mid  \lambda \in \Lambda_i, \; i \in \{1,2\} \}
\]
where $\pi_1 : \funT_1 \times \funT_2 \to \funT_1$ and $\pi_2 : \funT_1 \times \funT_2 \to \funT_2$ are the natural projection maps.
\end{definition}
For simplicity of notation we will allow a small bit of imprecision and just use
the symbol $\lambda$ to denote the predicate lifting $ \funQ \pi_i \circ
\lambda$, given a lifting $\lambda$ for $\funT_i$.
This means that we can view the languages $\MLone_{\Lambda_1}(A)$ and
$\MLone_{\Lambda_2}(A)$ as fragments of $\MLone_{\La_1 \otimes \La_2}(A)$.
Our goal is to prove the following.

\begin{thm}%
\label{t:comb-prod}
Suppose both signatures $\Lambda_1$ and $\Lambda_2$ admit a disjunctive basis.
Then so does $\Lambda_1 \otimes \Lambda_2 $.
\end{thm}

The disjunctive basis for the combined signature is defined as follows.

\begin{definition}
Let $\D_1,\D_2$ be disjunctive bases for $\Lambda_1,\Lambda_2$ respectively.
Given a set of variables $A$ we define the set of one-step formulas $(\D_1
\otimes \D_2)(A)$ to be all finite disjunctions of formulas of the form
$\delta_1 \wedge \delta_2$, where $\delta_1 \in \D_1(A)$ and
$\delta_2 \in \D_2(A)$.
\end{definition}

Note that $\top \in (\D_1
\otimes \D_2)(A)$, since it is in $\D_1(A)$ and $\D_2(A)$, given that we allow a slight abuse of notation and identify  the conjunction $\top \wedge \top$ with $\top$.

We will show that $\D_1 \otimes \D_2$ is a disjunctive basis for $\La_{1} \otimes
\La_{2}$ indeed.
The first thing we need to check is that these formulas are indeed disjunctive.

\begin{prop}
Let $\delta_1 \in \D_1(A)$ and $\delta_2 \in \D_2(A)$. Then the formula $\delta_1 \wedge \delta_2$ is disjunctive.
\end{prop}
\begin{proof}
Let $(X,(\xi_1,\xi_2),m)$ be a one-step $\funT_1 \times \funT_2$-model satisfying the formula $\delta_1 \wedge \delta_2$. Then $X,\xi_1,m \satone \delta_1$ and $X,\xi_2,m \satone \delta_2$. Since $\delta_1,\delta_2$ are disjunctive there exist two covering one-step frames $h_1 : (Y_1,\rho_1) \to (X,\xi_1)$ and $h_2 : (Y_2,\rho_2) \to (X,\xi_2)$ together with markings $m_1 : Y_1 \to \funP \Prop$ and $m_1 : Y_2 \to \funP \Prop$ such that for each $i \in \{1,2\}$:
\begin{itemize}
    \item $Y_i,\rho_i,m_i \satone \delta_i$,
    \item $m_i(u) \subseteq m(h_i(u))$ for all $u \in Y_i$,
    \item $\sz{m_i(u)} \leq 1$ for all $u \in Y_i$.
\end{itemize}
We shall construct a covering one-step $\funT_1 \times \funT_2$-frame based on
the co-product $Y_1 + Y_2$ of the sets $Y_1,Y_2$, with the covering map given
by the co-tuple of the maps $h_1,h_2$, as shown in Figure~\ref{coprod}.
\begin{figure}[h]
\[
\xymatrix{  Y_1 \ar^{i_1}[rr]\ar_{h_1}@(d,l)[ddrr] & &  Y_1 + Y_2  \ar_{[h_1,h_2]}@{.>}[dd] &  & Y_2  \ar@(d,r)^{h_2}[ddll] \ar_{i_2}[ll]\\ % chktex 3
& & & & \\
& & X & &
}
\]
\caption{The covering map $[h_1,h_2]$}\label{coprod}
\end{figure}
We define the covering one-step frame to be
$(Y_1 + Y_2, ((\funT_1 i_1) \rho_1, (\funT_2 i_2) \rho_2))$, where $i_1,i_2$ are the insertions of the co-product, so that $[h_1,h_2] \circ i_1 = h_1$ and $[h_1,h_2] \circ i_2 = h_2$. That the co-tuple $[h_1,h_2]$ is in fact a one-step frame morphism is shown by a straightforward calculation:
\begin{eqnarray*}
(\funT_1 \times \funT_2)[h_1,h_2](\funT_1 i_1 \rho_1, \funT_2 i_2 \rho_2) &  = & (\funT_1 [h_1,h_2] \circ \funT_1 i_1 (\rho_1), \funT_2 [h_1,h_2] \circ  \funT_2 i_2(\rho_2)) \\
& = & (\funT_1 ( [h_1,h_2] \circ  i_1) (\rho_1), \funT_2 ([h_1,h_2] \circ   i_2)(\rho_2)) \\
& = & (\funT_1 h_1(\rho_1), \funT_2 h_2  (\rho_2)) \\
& = & (\xi_1,\xi_2)
\end{eqnarray*}

The commutative diagram shown in Figure~\ref{bigdia} may help to give an overview of the construction, in which the top middle entry shows the type of the covering one-step frame we have constructed.

\begin{figure}[h]
\[
\xymatrix{   & \ar`l[dl] `d[dddd]_{\funT_1[h_1,h_2]} [dddd]    \funT_1 (Y_1 + Y_2) & \ar[l] \ar[r] (\funT_1 \times \funT_2)(Y_1 + Y_2)  \ar@{.>}_{(\funT_1 \times \funT_2)[h_1,h_2]}[dddd] & \funT_2 (Y_1 + Y_2)   \ar`r[dr] `d[dddd]^{\funT_2[h_1,h_2]} [dddd]   & \\ % chktex 3
&  & &  &  \\
& \funT_1 Y_1 \ar^{\funT_1 i_1}[uu] \ar[dd]_{\funT_1 h_1}  &  &  \funT_2 Y_2 \ar_{\funT_2 i_2}[uu] \ar^{\funT_2 h_2}[dd]  & \\ % chktex 3
& & & & & &  \\
& \funT_1 X  & \ar[l] (\funT_1 \times \funT_2)X \ar[r]  & \funT_2 X &
}
\]
\caption{The map $[h_1,h_2]$ is a one-step frame morphism}\label{bigdia}
\end{figure}

The construction is completed by defining a marking on the covering one-step $\funT_1 \times \funT_2$-frame  by co-tupling the markings for each covering $\funT_i$-frame, to obtain the marking $[m_1,m_2]$. Since each of the markings $m_1,m_2$ factor through the co-tuple $[m_1,m_2]$ via the insertions, we have $[m_1,m_2](u) \subseteq m([h_1,h_2](u)) $ for each $u \in Y_1 + Y_2$, and also
$\sz{[m_1,m_2](u)} \leq 1$ for each $u \in Y_1 + Y_2$.

By naturality of one-step formulas, and again since  the markings $m_1,m_2$ factor through the co-tuple $[m_1,m_2]$ via the insertions, we have
\[
    Y_1 + Y_2, \funT_1 i_1 (\rho_1), [m_1,m_2] \satone \delta_1
\]
and
\[
    Y_1 + Y_2,  \funT_2 i_2 (\rho_2), [m_1,m_2] \satone \delta_2
\]
It easily follows that:
\[
    Y_1+Y_2, ( \funT_1 i_1 (\rho_1), \funT_2 i_2 (\rho_2)),[m_1,m_2] \satone \delta_1 \wedge \delta_2
\]
as required. So $\delta_1 \wedge \delta_2$ is disjunctive.
\end{proof}

\begin{proofof}{Theorem~\ref{t:comb-prod}}
Since we know the formulas in $(\D_1 \otimes \D_2)(A)$ are disjunctive, we only
need to check conditions (1) -- (3) of Definition~\ref{d:disbas} one by one: % chktex 8
condition (1) just requires the formulas in a disjunctive bases to contain
$\top$ and be closed under disjunctions, which holds by definition of
$\D_1 \otimes \D_2$.
For condition (2), consider a one-step formula in $\MLone_{\La_1 \otimes \La_2}(A)$ of the form
$\hs_{\la}\ol{\pi}$ where $\la \in \La_1$ or $\la \in \La_2$. Suppose the former is the case.  Then there is a  formula $\delta \in \D_1(\funP A)$ such that
 such that $\hs_{\la} \ol{\pi} \oneseq \delta[\land_A]$. But then $\delta \wedge \top \in (\D_1 \otimes \D_2)(A)$, and this formula is also one-step equivalent to $\hs_{\la} \ol{\pi}$.

Finally, for condition (3), by rewriting positive one-step formulas into disjunctive normal form (treating modalities as atomic)  we only need to consider conjunctions of the form:
\[
    (\delta_1 \wedge \delta_2) \wedge (\delta_1' \wedge \delta_2')
\]
where $\delta_1 \in \D_1(A)$, $\delta_2 \in \D_2(A)$, $\delta_1' \in \D_1(B)$, $\delta_2' \in \D_2(B)$. Apply the distributive law of $\D_1$ to $\delta_1 \wedge\delta_1'$ and that of $\D_2$ to $\delta_2 \wedge \delta_2'$ to find formulas $\gamma_1 \in \D_1(A \times B)$ and $\gamma_2 \in \D_2(A \times B)$ such that $\delta_1  \wedge \delta_1' \oneseq \gamma_1[\theta_{A,B}]$ and $\delta_2 \wedge \delta_2' \oneseq \gamma_2[\theta_{A,B}]$. The conjunction $\gamma_1 \wedge \gamma_2$ is in $(\D_1 \otimes \D_2)(A \times B)$ and we have:
\[
    (\delta_1 \wedge \delta_2) \wedge (\delta_1' \wedge \delta_2') \oneseq (\gamma_1 \wedge \gamma_2)[\theta_{A,B}]
\]
as required.
\end{proofof}

\subsection{Composition}

The third and final example that we consider involves the \emph{composition}
$\funT = \funT_{1} \circ \funT_{2}$ of two set functors $\funT_{1}$ and
$\funT_{2}$.
We first recall C\^{\i}rstea \& Pattinson's definition of the combined
signature in the case of composition~\cite{cirs:modu07}.

\begin{definition}
Let $\la \in \La_{1}$ be an $m$-ary predicate lifting, and let $\al_{1}, \ldots,
\al_{m}$ be one-step formulas in  $\MLone_{\La_{2}}(A)$ for some set
$A = \{ a_{1}, \ldots, a_{n} \}$.
Then we define the $n$-ary predicate lifting
\[
\ul{\la\tup{\al_{1},\ldots,\al_{m}}}: \funQ^{n} \Rightarrow \funQ\funT
\]
in the obvious way:
\[
\ul{\la\tup{\al_{1},\ldots,\al_{m}}}_{S}:
(U_{1},\ldots,U_{n}) \mapsto \la_{\funT_{2}S}
\Big(
\wh{\al_{1}}_{S}(U_{1},\ldots,U_{n}),\ldots,\wh{\al_{m}}_{S}(U_{1},\ldots,U_{n})
\Big).
\]
We let $\La_{1}\cmp\La_{2}$ denote the set of all (finitary) predicate liftings
that can be obtained in this way.
\end{definition}

In the sequel we will often abbreviate $\La_{1}\cmp\La_{2}$ as $\La$.

\begin{thm}%
\label{t:cmp}
Suppose both signatures $\Lambda_1$ and $\Lambda_2$ admit a disjunctive basis.
Then so does $\Lambda_1 \cmp \Lambda_2 $.
\end{thm}

For the definition of the disjunctive basis of the combined signature, fix a set
$A = \{ a_{1}, \ldots, a_{n} \}$ and consider the one-step language
$\MLone_{\La_{1}}(\MLone_{\La_{2}}(A))$.
Given a formula $\al$ in this language, every occurrence of an $m$-ary
$\La_{1}$-modality $\hs_{\la}$ is of the form $\hs_{\la}(\ga_{1},\ldots,\ga_{m})$
with each $\ga_{i} \in \MLone_{\La_{2}}(A)$.
If we now replace each such subformula $\hs_{\la}(\ga_{1},\ldots,\ga_{m})$ with
the formula
$\hs_{\ul{\la\tup{\ga_{1},\ldots,\ga_{m}}}}(a_{1},\ldots,a_{n})$,
we have associated with $\al$ a unique formula $\al'\in \MLone_{\La}(A)$.

% \btbs
% \item
% Below we shall use the following semantic link between the formulas $\al$ and
% $\al'$.
% \etbs

\begin{definition}
Define
\[
(\D_{1}\cmp\D_{2})(A) \isdef \{ \de' \in \MLone_{\La}(A) \mid \de \in \D_{1}(\D_{2}(A)) \}.
\]
where $\D_{1}$ and $\D_{2}$ are disjunctive bases for $\La_{1}$ and $\La_{2}$,
respectively.
\end{definition}

To prove that $\D_{1}\cmp\D_{2}$ is a disjunctive basis for the signature
$\Lambda_1 \cmp \Lambda_2$ we first show that $\D_{1}\cmp\D_{2}(A)$ consists of
disjunctive formulas, for any set $A$.

\begin{prop}
Every formula in $\D_{1}\cmp\D_{2}(A)$ is disjunctive.
\end{prop}

\begin{proof}
It suffices to show that for an arbitrary but fixed formula $\de \in
\D_{1}(\D_{2}(A))$, the formula $\de' \in (\D_{1}\cmp\D_{2})(A)$ is disjunctive.

Let $(S,\si,m)$ be an arbitrary one-step $\funT$-model such that $S,\si,m
\satone \de^{1}$, and let $\de \in \D_{1}(B)$ and $\eta: B \to \D_{2}(A)$ be such
that $\de = \de_{1}[\eta]$.
Consider the one-step $\funT_{1}$-model $(\funT_{2}S,\si,m_{\eta})$ where
$m_{\eta}: \funT_{2}S \to \funP B$ is given by
\[
m_{\eta}(\rho) \isdef \{ b \in B \mid (S,\rho,m) \satone \eta(b) \}.
\]

% \begin{claimfirst}
% \label{cl:cmp1}
% Let $\al \in \MLone_{\La_{1}}(\MLone_{\La_{2}}(A))$ be of the form $\al =
% \be[\tau]$ for some formula $\be \in \MLone_{\La_{1}}(B)$ and substitution
% $\tau: B \to \MLone_{\La_{2}}(B)$.
% Then for any one-step $\funT$-model $(S,\si,m: S \to \funP A)$ we have
Unravelling the definitions, it is not hard to show that for any formula
$\al_{1} \in \MLone_{{\La}_{1}}(B)$ we have
\begin{equation}
\label{eq:cmp0}
(S,\si,m) \satone \al_{1}[\eta] \text{ iff }
(\funT_{2}S,\si,m_{\eta}) \satone \al_{1},
\end{equation}
% where $(\funT_{2}S,\si,m_{\tau})$ is the $\funT_{1}$-model, with domain
% $\funT_{2}S$, and marking $m_{\tau}: \funT_{2}S \to B$ given by $m_{\tau}(\rho)
% \isdef \{ b\in B\mid (S,\rho,m) \satone \tau_{b} \}$.
% \end{claimfirst}
and as an immediate consequence of this we find that
\[
(\funT_{2},\si,m_{\eta}) \satone \delta_{1}.
\]
Then by disjunctiveness of $\D_{1}$ there is a one-step $\funT_{1}$-model
$(Z, \zeta, n)$, where $\zeta \in \funT_{1}Z$ and $n: Z \to \funP(B)$,
together with a map $f: Z \to \funT_{2}S$ such that
$(\funT_{1}f)\zeta = \sigma$;
$\sz{n(z)} \leq 1$ and $n(z) \sse m_{\eta}(f(z))$, for all $z \in Z$; and
%    \textcolor{red}{for the time being assume $\sz{n(z)} = 1$}
$(Z,\zeta,n) \satone \delta_{1}$.

Define, for $z \in Z$, $\eta_{z} \in \D(A)$ to be the formula $\eta(b)$ in
case $b \in B$ is the unique element of $n(z)$; and set $\eta_{z} \isdef \top$
if $n(z) = \nada$.
Observe that for any $z \in Z$, the triple $(S,f(z),m)$ constitutes a
$\funT_{2}$-model; we claim that
\begin{equation}
\label{eq:cmp1}
(S,f(z),m) \satone \eta_{z}.
\end{equation}
To see this, clearly we only have to pay attention to the case where $n(z) =
\{b\}$ for some $b \in B$.
But here it is immediate from $n(z) \sse m_{\eta}(f(z))$ that $b \in
m_{\eta}(f(z))$, and so we obtain~\eqref{eq:cmp1} by definition of $m_{\eta}$.

Given~\eqref{eq:cmp1}, we now use the disjunctiveness of $\D_{2}$ to obtain,
for each $z \in Z$, a one-step $\funT_{2}$-model $(S_{z},\rho_{z},m_{z})$,
with $\rho_{z} \in \funT_{2}S_{z}$ and $m_{z}: S_{z} \to \funP(A)$,
together with a map $g_{z}: S_{z} \to S$, such that
$(T_{2}g_{z})(\rho_{z}) = f(z)$;
$\sz{m_{z}} \leq 1$ and $m_{z}(t) \sse m(f(t))$, for all $t \in S_{z}$;
and $(S_{z},\rho_{z},m_{z}) \satone \eta_{z}$.

We are now ready to define the required separating cover for $(S,\si,m)$.
As its domain we will take the coproduct $\coprod_{z\in Z} S_{z}$, and where
for $z \in Z$ we let $i_{z}: S_{z} \to S'$ denote the insertion map, we define
the maps $m': S' \to \funP A$ and $g: S' \to S$ via co-tupling; in particular,
we require $m \circ i_{z} = m_{z}$ and $g \circ i_{z} = g_{z}$ for all $z \in Z$.
For the definition of the distinguished element $\si' \in \funT S'$, we first
define the map $\rho': Z \to \funT_{2} S'$ by setting $\rho'(z) \isdef
(\funT_{2} i_{z})\rho_{z}$.
We then put $\si'\isdef (\funT_{1} \rho') \zeta$.

It is obvious that $(S',\si',m)$ is a one-step $\funT$-model; we now check that
together with the map $g: S' \to S$ is indeed a separating cover for
$(S,\si,m)$.
First of all, it is obvious that $\sz{m'(s')} \leq 1$ and $m'(z) \sse m(g(z))$,
for all $s' \in S'$.
Second, to check that
\begin{equation}
\label{eq:cmp2}
(\funT g)\si' = \si
\end{equation}
we first observe that $(\funT_{2}g) \circ \rho' = f$, as can easily verified:
$((\funT_{2}g)\circ \rho')(z) = (\funT_{2}g)(\rho'(z)) =
(\funT_{2}g)((\funT_{2} i_{z})\rho_{z}) =
(\funT_{2} (g \circ i_{z}))(\rho_{z}) =
(\funT_{2} g_{z})(\rho_{z}) = f(z)$.
We then continue with the following calculation:
\[
(\funT g)\si' = (\funT_{1} \funT_{2} g) ((\funT_{1} \rho')\zeta)
= ((\funT_{1} \funT_{2} g) \circ (\funT_{1} \rho'))(\zeta)
= (\funT_{1}((\funT_{2}g) \circ \rho'))(\zeta)
= (\funT_{1}f)(\zeta)
= \si.
\]

Finally, we need to prove that
\begin{equation}
\label{eq:cmp3}
(S',\si',m') \satone \de'.
\end{equation}
To see this, let $m'_{\eta}: \funT_{2}S' \to \funP B$ be the marking given by
\[
m'_{\eta}(\rho') \isdef \{ b \in B \mid (S',\rho',m') \satone \eta(b) \},
\]
and consider the marking $n': Z \to \funP B$ defined by $n'(z) \isdef
m'_{\eta}(\rho'_{z})$.
Our key claim is now that $n \sse n'$, and we prove this as follows.
In case $n(z) = \nada$ there is nothing to prove; in case $n(z) \neq \nada$,
let $b \in B$ be the unique element of $n(z)$.
Then we have
\begin{align*}
b \in n(z) & \Rightarrow b \in m_{\eta}(f(z))
   & \text{(assumptions on $f$ and $n$)}
\\ & \Rightarrow (S, f(z), m) \satone \eta(b)
   & \text{(definition $m_{\eta}$)}
\\ & \Rightarrow (S_{z}, \rho_{z}, m_{z}) \satone \eta(b)
   & \text{($(S_{z}, \rho_{z}, m_{z})$ is cover)}
\\ & \Rightarrow (S', \rho'_{z}, m') \satone \eta(b)
   & \text{(invariance under $i_{z}$, Prop.~\ref{p:1invar})}
\\ & \Rightarrow b \in m'_{\eta}(\rho'_{z})
   & \text{(definition $m'_{\eta}$)}
\\ & \Rightarrow b \in n'(z)
   & \text{(definition $n'$)}
\end{align*}
But by the monotonicity of disjunctive formulas, it follows from $n \sse n'$ and
$(Z,\zeta,n) \satone \delta_{1}$ that $(Z,\zeta,n') \satone \delta_{1}$.
Then by invariance (Proposition~\ref{p:1invar}) we find that $(\funT_{2}S,\si,
m'_{\eta}) \satone \delta_{1}$, and from this we may derive~\eqref{eq:cmp3},
using an analogous claim to~\eqref{eq:cmp0}.
\end{proof}

\begin{proofof}{Theorem~\ref{t:cmp}}
Since we have verified the disjunctivity of all formulas in
$\D_{1}\cmp\D_{2}(A)$, it remains to check that $\D_{1}\cmp\D_{2}(A)$ satisfies
the closure conditions of disjunctive bases.
For condition (1) this is an immediate consequence of the definitions.
It is in fact not very hard to see that $\D_{1}\cmp\D_{2}$ meets the other two
closure conditions as well, but full proofs are very tedious.
In order to avoid convoluted syntax we confine ourselves to somewhat sketchy
arguments here.

For condition (2), consider a formula of the form
$\hs_{\ul{\la\tup{\al_{1},\ldots,\al_{m}}}}(b_{1},\ldots,b_{n})[\pi]$,
where $\la \in \La_{1}$, each $\al_{i} \in \MLone_{\La_{2}}(B)$ and $\pi: B \to
\Latt(A)$.
Since $\D_{1}$ is a disjunctive basis for $\La_{1}$, we may find formulas
$\de_{1} \in \D_{1}(\{ 1,\ldots,k\})$ and $\be_{1},\ldots,\be_{k} \in
\MLone_{\La_{2}}(A)$ such that $\hs_{\la}(\al_{1}[\pi],\ldots,\al_{m}[\pi])
\equiv^{1} \de_{1}(\be_{1},\ldots,\be_{k})$.
But, now using the fact that $\D_{2}$ is a disjunctive basis for $\La_{2}$,
we may derive from Proposition~\ref{p:db} that each formula $\be_{i}$ is
equivalent to a formula $\ga_{i}[\wedge_{A}]$, where $\ga_{i} \in
\D_{2}(\funP A)$.
It is then a tedious but straightforward exercise to show that
$\hs_{\ul{\la\tup{\al_{1},\ldots,\al_{m}}}}(b_{1},\ldots,b_{n})[\pi]
\equiv^{1} \de'$, where we define $\de \isdef \de_{1}(\ga_{1},\ldots,\ga_{k})$.

Finally, for condition (3), consider two disjunctive formulas
$\ga' \in \D_{1}\cmp\D_{2}(A)$ and $\de' \in \D_{1}\cmp\D_{2}(B)$,
where $\ga = \ga_{1}[\si]$ and $\de = \de_{1}[\tau]$
for $\ga_{1} \in \D_{1}(A')$, $\de_{1} \in \D_{1}(B')$,
and $\si: A' \to \D_{2}(A)$, $\tau: B' \to \D_{2}(B)$
for some sets $A'$ and $B'$ that without loss of generality we may
take to be disjoint.
$\D_{1}$ being a disjunctive basis yields a formula $\be_{1} \in \D(A' \utimes
B')$ such that $\ga_{1} \land \de_{1} \equiv^{1} \be_{1}[\theta_{A',B'}]$.
By the disjointness of $A'$ and $B'$ we then have
\begin{equation}
\label{eq:cmp11}
\ga' \land \de' \equiv^{1} \be_{1}[\theta_{A',B'}][\si][\tau].
\end{equation}
Now consider an arbitrary pair $(a',b') \in A' \times B'$; since $\D_{2}$
is a disjunctive basis there is a formula $\al_{a',b'} \in \D_{2}(A\times B)$
such that $\si_{a'}\land \tau_{b'} \equiv^{1} \al_{a',b'}[\theta_{A,B}]$.
Define the following substitution $\si \utimes \tau: A' \utimes B' \to \
D_{2}(A' \utimes B')$:
\[
\si\utimes\tau(d) \isdef
   \begin{cases}
      \si(d)      & \text{ if } d \in A'
   \\ \tau(d)     & \text{ if } d \in B'
   \\ \al_{a',b'} & \text{ if } d = (a',b') \in A' \times B'
   \end{cases}
\]
It then follows from the definitions that the substitutions
$[\theta_{A',B'}][\si][\tau]$ and $[\si\utimes\tau][\theta_{A,B}]$ produce
one-step equivalent formulas, so that combining this with~\eqref{eq:cmp11} we
obtain that
\[
\ga' \land \de' \equiv^{1} \be_{1}[\si\utimes\tau][\theta_{A,B}].
\]
This suffices, since obviously we have that $\be_{1}[\si\utimes\tau]
\in \D_{1}\cmp\D_{2}(A\utimes B)$.
\end{proofof}

%%% Local Variables:
%%% mode: latex
%%% TeX-master: "main.tex"
%%% End:

%% file: sec-yoneda.tex
\section{Yoneda representation of disjunctive liftings}%
\label{s:yoneda}

It is a well known fact in coalgebraic modal logic that predicate liftings have
a neat representation via an application of the Yoneda lemma.
This was explored by Schr\"{o}der in~\cite{schr:expr08}, where it was used among
other things to prove a characterization theorem for the monotone predicate
liftings.
Here, we apply the same idea to disjunctive liftings. We shall be working with a slightly generalized notion of predicate lifting here, taking a predicate lifting over a finite set of variables $A$ to be a natural transformation  $\lambda : \funQ^A \to \funQ \circ \funT$. Clearly, one-step formulas in $\MLoneLa(A)$ can then be viewed as predicate liftings over $A$.

\begin{definition}
Let $\lambda : \funQ^A \to \funQ \circ \funT$ be a predicate lifting over
variables $A = \{a_1,\dots,a_n\}$.
The \emph{Yoneda representation} $y(\lambda)$ of $\lambda$ is the subset
\[
\lambda_{\funP A}(\mathsf{true}_{a_1},\dots,\mathsf{true}_{a_n}) \in
\funP \funT \funP A
\]
where $\mathsf{true}_{a_i} = \{B \subseteq A \mid a_i \in B\}$.
We shall write simply $\lambda \subseteq \funT \funP A$ instead of $y(\lambda)$.
\end{definition}

\begin{definition}
Given a set $A$, let $A^\top$ be the set $A \cup \{\top\}$.
Let $\epsilon_A \subseteq  A^\top \times \funP A$ be the relation defined by
$a \epsilon_A B$ iff $a \in B$, and $\top \epsilon_A B$ for all $B \subseteq A$.
Let $\eta_A : A^\top \to \funP A$ be defined by $\eta_A(a) = \{a\}$, and
$\eta_A(\top) = \nada$.
\end{definition}

In the remainder of this section we assume familiarity with the Barr relation
lifting $\ol{\funT}$ associated with a functor $\funT$; see~\cite{kupk:comp12}
for the definition and some basic properties.

\begin{definition}
A predicate lifting $\lambda \subseteq \funT \funP A$ is said to be
\emph{divisible} if, for all $\alpha \in \lambda$ there is some $\beta \in \funT
A^\top$ such that
$(\beta,\alpha) \in \overline{\funT}(\epsilon_A)$ and
$\fun \eta_A (\beta) \in \lambda$.
\end{definition}

\begin{prop}%
\label{p:yoneda-char}
Any disjunctive lifting over $A$ is divisible, and if $\funT$ preserves weak
pullbacks the disjunctive liftings over $A$ are precisely the divisible ones.
\end{prop}

% \section{Proof of Proposition \ref{p:yoneda-char}}
\begin{proof}
Suppose $\lambda \subseteq \funT \funP A$ is disjunctive, and pick $\alpha \in
\lambda$.
Then $\funP A,\alpha,\id_{\funP A} \satone \lambda$, so since
$\lambda$ is disjunctive there are some one-step model $(X,\xi,m)$ and
map $f: X \to \funP A$ with $m : X \to \funP A$,
$m(u) \subseteq f(u)$ for all $u \in X$, $\fun f (\xi) = \alpha$, and
$\sz{m(u)} \leq 1$ for all $u \in X$.
We define a map $g : X \to A^\top$ by setting $g : u \mapsto \top$ if $m(u) =
\nada$, $g : u \mapsto a$ if $m(u) = \{a\}$.
We tuple the maps $f,g$ to get a map $\langle f,g \rangle : X \to A^\top \times
\funP A $.
In fact, since $m(u) \subseteq f(u)$ for all $u \in X$, we have $\langle f,g
\rangle : X \to \epsilon_A$. Let $\pi_1 : \epsilon_A \to A^\top$ and $\pi_2 :
\epsilon_A \to \funP A$ be the projection maps.
We have the following diagram, in which the two triangles and the outer edges
commute (i.e., $m = \eta_{A}\circ g$).
\[
\xymatrix{ & & \funP A \\
X \ar@(u,l)^{m}[urr]\ar^{f}[urr]\ar_{g}[drr] \ar_{\langle f,g \rangle}[rr] & & \epsilon_A \ar_{\pi_2}[u] \ar^{\pi_1}[d] \\ % chktex 3
& & A^\top \ar@(r,r)_{\eta_A}[uu] } % chktex 3
\]
Now apply $\funT$ to this diagram and define $\beta \in \funT A^\top$ to be
$\funT(\pi_1 \circ \langle f,g\rangle)(\xi) = \funT g(\xi)$.
First, we have $(\beta,\alpha) \in \overline{\funT}(\epsilon_A)$, witnessed by
$\funT(\langle f,g\rangle)(\xi) \in \funT \epsilon_A$.
We claim that $\funT \eta_A(\beta) \in \lambda$.
But since $X,\xi,m \satone \lambda$ and $m = \eta_A \circ g$, naturality of
$\lambda$ applied to the map $g : X \to A^\top$, gives $A^\top,\beta,\eta_A
\satone \lambda$.
% Another naturality argument, applied to the map $\eta_A : A^\top \to \funP A$
% and using the fact that $\eta_A \circ \id_{\funP A} = \eta_A$, gives $\funP A,
% \funT \eta_A (\beta), \id_{\funP A} \satone \lambda$, i.e.,
% $\funT \eta_A (\beta) \in \lambda$.
Another naturality argument, applied to
% the one-step model morphism
$\eta_A : (A^\top,\be,\eta_{A}) \to (\funP A,\funT \eta_A (\be), \id_{\funP A})$
% and using the fact that $\eta_A \circ \id_{\funP A} = \eta_A$,
gives $\funP A,\funT \eta_A (\beta), \id_{\funP A} \satone \lambda$, i.e.,
$\funT \eta_A (\beta) \in \lambda$.

For the converse direction, under the assumption that $\funT$ preserves weak
pullbacks, suppose that $\lambda$ is divisible, and suppose $X,\xi,m \satone
\lambda$. We get $\funT m(\xi) \in \lambda$ and so we find some $\beta \in
\fun A^\top$ with $\beta (\overline{\funT} \epsilon_A) \funT m (\xi) $ and
$\funT \eta_A(\beta) \in \lambda$.
Pick some $\beta'\in \funT \epsilon_A$ with $\funT \pi_2 (\beta') = \funT m(\xi)$
and $\funT\pi_1 (\beta') = \beta$.
Let $R,g_1,g_2$ be the pullback of the diagram $X \rightarrow \funP A \leftarrow
\epsilon_A$, shown in the diagram:

\[
    \xymatrix{ X \ar^m[rr] & & \funP A \\
R\ar^{g_1}[u] \ar_{g_2}[rr] & & \epsilon_A \ar_{\pi_2}[u] \ar^{\pi_1}[d] \\
& & A^\top \ar@(r,r)_{\eta_A}[uu] } % chktex 3
\]

By weak pullback preservation there is $\rho \in \funT R$ with $\funT g_1(\rho)
= \xi$ and $\funT g_2 (\rho) = \beta'$.
The map $g_1 : (R,\rho) \to (X,\xi)$ is thus a cover, and we have a marking
$m'$ on $R$ defined by $\eta_A \circ \pi_1 \circ g_2$ (follow the bottom-right
path in the previous diagram).
It is now routine to check that $R,\rho,m' \satone \lambda$, and $\sz{m'(u)}
\leq 1$ and $m'(u) \subseteq m(g_1(u))$ for all $u \in R$, so we are done.
\end{proof}

For the moment, we leave the question open, whether a similar characterization
of disjunctive predicate liftings can be proved without weak pullback
preservation.
We also leave it as an open problem to characterize the functors that admit
a disjunctive basis.

%%% Local Variables:
%%% mode: latex
%%% TeX-master: "main.tex"
%%% End: